\documentclass[11pt,letterpaper]{article}
\usepackage[margin=1in]{geometry}

\usepackage[expansion=false]{microtype} 
\usepackage{graphicx}
\usepackage{subcaption}
\usepackage{booktabs}
\usepackage{makecell}
\usepackage{amsmath}
\usepackage{amssymb}
\usepackage{amsbsy}
\usepackage{mathtools}
\usepackage{amsthm}
\DeclareFontShape{OT1}{cmr}{m}{scit}{<->ssub*cmr/m/sc}{}
\IfFileExists{bbm.sty}{\usepackage{bbm}}{\providecommand{\mathbbm}[1]{\mathbb{##1}}}
\IfFileExists{dsfont.sty}{\usepackage{dsfont}}{}
\usepackage{algorithm}
\IfFileExists{algpseudocode.sty}{%
  \usepackage{algpseudocode}%
}{%
  \usepackage{algorithmic}%
  \newcommand{\State}{\STATE}%
  \newcommand{\For}[1]{\FOR{##1}}%
  \newcommand{\EndFor}{\ENDFOR}%
  \newcommand{\If}[1]{\IF{##1}}%
  \newcommand{\Else}{\ELSE}%
  \newcommand{\EndIf}{\ENDIF}%
  \newcommand{\While}[1]{\WHILE{##1}}%
  \newcommand{\EndWhile}{\ENDWHILE}%
  \newcommand{\Require}{\REQUIRE}%
  \newcommand{\Ensure}{\ENSURE}%
}
\providecommand{\REQUIRE}{\Require}
\providecommand{\ENSURE}{\Ensure}
\providecommand{\STATE}{\State}

\providecommand{\FOR}[1]{\For{#1}}
\providecommand{\ENDFOR}{\EndFor}
\providecommand{\IF}[1]{\If{#1}}
\providecommand{\ELSE}{\Else}
\providecommand{\ENDIF}{\EndIf}
\providecommand{\WHILE}[1]{\While{#1}}
\providecommand{\ENDWHILE}{\EndWhile}
\usepackage{xcolor}
\usepackage{nicefrac}
\usepackage{esint}
\usepackage{comment}
\usepackage{ifthen}
\usepackage{float}
\usepackage{url}
\usepackage{bm}
\usepackage{colortbl}
\usepackage{enumitem}
\usepackage{prettyref} 
\usepackage{tikz}
\usetikzlibrary{positioning,arrows,decorations,decorations.markings,decorations.pathreplacing,decorations.pathmorphing,shadows,positioning,arrows.meta,matrix,fit}
\newrefformat{eq}{\savehyperref{#1}{\textup{(\ref*{#1})}}}
\newrefformat{eqn}{\savehyperref{#1}{Equation~\ref*{#1}}}
\newrefformat{lem}{\savehyperref{#1}{Lemma~\ref*{#1}}}
\newrefformat{def}{\savehyperref{#1}{Definition~\ref*{#1}}}
\newrefformat{line}{\savehyperref{#1}{line~\ref*{#1}}}
\newrefformat{thm}{\savehyperref{#1}{Theorem~\ref*{#1}}}
\newrefformat{corr}{\savehyperref{#1}{Corollary~\ref*{#1}}}
\newrefformat{sec}{\savehyperref{#1}{Section~\ref*{#1}}}
\newrefformat{ssec}{\savehyperref{#1}{Subsection~\ref*{#1}}}
\newrefformat{app}{\savehyperref{#1}{Appendix~\ref*{#1}}}
\newrefformat{ass}{\savehyperref{#1}{Assumption~\ref*{#1}}}
\newrefformat{ex}{\savehyperref{#1}{Example~\ref*{#1}}}
\newrefformat{fig}{\savehyperref{#1}{Figure~\ref*{#1}}}
\newrefformat{alg}{\savehyperref{#1}{Algorithm~\ref*{#1}}}
\newrefformat{rem}{\savehyperref{#1}{Remark~\ref*{#1}}}
\newrefformat{conj}{\savehyperref{#1}{Conjecture~\ref*{#1}}}
\newrefformat{prop}{\savehyperref{#1}{Proposition~\ref*{#1}}}
\newrefformat{proto}{\savehyperref{#1}{Protocol~\ref*{#1}}}
\newrefformat{prob}{\savehyperref{#1}{Problem~\ref*{#1}}}
\newrefformat{claim}{\savehyperref{#1}{Claim~\ref*{#1}}}
\newrefformat{fact}{\savehyperref{#1}{Fact~\ref*{#1}}}
\usepackage[round,authoryear]{natbib}

\usepackage[breaklinks]{hyperref}
\makeatletter
\providecommand{\theHALG@line}{}
\renewcommand{\theHALG@line}{\thealgorithm.\arabic{ALG@line}}
\makeatother
\usepackage{cleveref}

\crefname{equation}{}{}
\crefname{proposition}{Proposition}{Propositions}
\crefname{appendix}{Appendix}{Appendices}
\crefname{assumption}{Assumption}{Assumptions}

\theoremstyle{plain}
\newtheorem{theorem}{Theorem}[section]
\newtheorem{proposition}[theorem]{Proposition}
\newtheorem{lemma}[theorem]{Lemma}

\theoremstyle{definition}

\newtheorem{example}[theorem]{Example}
\newtheorem{assumption}[theorem]{Assumption}
\theoremstyle{remark}

\makeatletter
\def\ddefloop#1{\ifx\ddefloop#1\else\ddef{#1}\expandafter\ddefloop\fi}
\def\ddef#1{\expandafter\def\csname bb#1\endcsname{\ensuremath{\mathbb{#1}}}}
\ddefloop ABCDEFGHIJKLMNOPQRSTUVWXYZ\ddefloop
\def\ddefloop#1{\ifx\ddefloop#1\else\ddef{#1}\expandafter\ddefloop\fi}
\def\ddef#1{\expandafter\def\csname b#1\endcsname{\ensuremath{\mathbf{#1}}}}
\ddefloop ABCDEFGHIJKLMNOPQRSTUVWXYZ\ddefloop
\def\ddef#1{\expandafter\def\csname c#1\endcsname{\ensuremath{\mathcal{#1}}}}
\ddefloop ABCDEFGHIJKLMNOPQRSTUVWXYZ\ddefloop
\def\ddef#1{\expandafter\def\csname h#1\endcsname{\ensuremath{\widehat{#1}}}}
\ddefloop ABCDEFGHIJKLMNOPQRSTUVWXYZ\ddefloop
\def\ddef#1{\expandafter\def\csname hc#1\endcsname{\ensuremath{\widehat{\mathcal{#1}}}}}
\ddefloop ABCDEFGHIJKLMNOPQRSTUVWXYZ\ddefloop
\def\ddef#1{\expandafter\def\csname t#1\endcsname{\ensuremath{\widetilde{#1}}}}
\ddefloop ABCDEFGHIJKLMNOPQRSTUVWXYZ\ddefloop
\def\ddef#1{\expandafter\def\csname tc#1\endcsname{\ensuremath{\widetilde{\mathcal{#1}}}}}
\ddefloop ABCDEFGHIJKLMNOPQRSTUVWXYZ\ddefloop
\makeatother

\let\Pr\undefined

\DeclareMathOperator{\Pr}{Pr}
\DeclareMathOperator*{\argmin}{arg\,min} 
\DeclareMathOperator*{\argmax}{arg\,max}

\newcommand{\ind}{\mathbbm{1}}    

\newcommand{\veps}{\varepsilon}

\newcommand\norm[1]{\left\lVert#1\right\rVert}
\newcommand{\E}{\mathbb{E}}

\newcommand{\R}{\mathbb{R}}

\newcommand{\ba}{\begin{array}}
\newcommand{\ea}{\end{array}}
\newcommand{\bs}{\begin{align}\begin{split}\nonumber}
\newcommand{\bsnumber}{\begin{align}\begin{split}}
\newcommand{\es}{\end{split}\end{align}}

\newcommand{\opnorm}[1]{\norm{#1}_{op}} %
\newcommand{\fronorm}[1]{\norm{#1}_{F}} %

\def\balign#1\ealign{\begin{align}#1\end{align}}
\def\balignat#1\ealign{\begin{alignat}#1\end{alignat}}
\def\bitemize#1\eitemize{\begin{itemize}#1\end{itemize}}
\def\benumerate#1\eenumerate{\begin{enumerate}#1\end{enumerate}}
\def\balignt#1\ealignt{\begin{talign}#1\end{talign}}%
\newcommand{\qb}{\mathbf{q}}

\newcommand{\PP}{\mathbb{P}} 
\newcommand*{\thetab}{\bm{\theta}}
\newcommand*{\taub}{\bm{\tau}}

\newcommand{\savehyperref}[2]{\texorpdfstring{\hyperref[#1]{#2}}{#2}}

\title{Treatment Effect Modification by Latent Classes}
\author{Hui Lan \quad Nathan Kallus \quad Winston Chou\\[0.4em]Netflix}
\date{}

\begin{document}
\maketitle

\begin{abstract}
Many treatments have effects that plausibly vary by latent classes. For example, the effect of personalized video recommendations may be stronger for users who seek a new discovery than for users seeking to rewatch something familiar. In the context of user journeys on digital platforms, we can conceptualize user moods as dynamic latent states that lead to distinct action distributions. Because mood is unobserved, we infer it from a user's pre-treatment actions using a hidden Markov model (HMM), and then leverage the latent state as an interpretable treatment effect modifier when evaluating experiment results. Under a treatment-effect sufficiency assumption, each observable conditional effect is a known mixture of the latent effects, so state-specific effects are identified when state posteriors are sufficiently heterogeneous. We propose an estimator that is Neyman-orthogonal to errors in the fitted HMM, enabling standard influence-function confidence intervals. In a real-world large-scale A/B test of a new recommendation model at Netflix, our latent-state analysis reveals that heavy-engagement users show significantly positive effects across all treatment variants, with limited evidence of benefit among less-active and rewatch-oriented users.
\end{abstract}

\section{Introduction}

Estimating heterogeneous treatment effects (HTE) is essential for optimizing personalized interventions. While standard approaches typically condition on observable covariates, treatment effects in many real-world applications are fundamentally modified by unobserved latent classes. For example, despite an insignificant overall mortality effect, \citet{famous2017ards} identified distinct treatment responses across latent classes of acute respiratory distress syndrome. As another example, \citet{lyu2023estimating} found that educational interventions varied across latent student types. Beyond static classes, the relevant latent context often evolves over time. In clinical data, neurostimulation responses have been linked to transient seizure-risk states inferred via hidden Markov models (HMMs) \citep{chiang2021state}. Similarly, HMMs have been applied to track evolving buyer--seller relationships to evaluate state-dependent marketing strategies \citep{zhang2016dynamic}. Other applications of dynamic states include digital health \citep{nahumshani2018jitai,boruvka2018assessing} and student learning \citep{corbett1995knowledge}. These examples highlight the need to understand treatment effect heterogeneity with respect to latent classes and/or dynamic states.

This work is motivated by a large-scale industry application at Netflix. A core hypothesis driving this work is that users' behavior at Netflix ebbs and flows according to dynamic user moods. For instance, users can shift between a discovery state, where they are actively looking for new content, and a nostalgic state, where they default to familiar content that they have watched before. In turn, these states should shape a user's receptiveness to algorithmic recommendations.

We propose a two-step framework that decouples latent class learning from causal effect estimation.  Specifically, we train a latent class model on pre-treatment data to infer underlying states.  At Netflix, we train an HMM on pre-treatment user behavioral histories using an independent sample of users disjoint from the experimental cohort to infer dynamic latent states. Leveraging the resulting state posteriors, we can mathematically unmix the observable conditional effects to estimate the true average treatment effect (ATE) within each latent state. 

Our contributions are threefold:
\begin{itemize}
    \item \textbf{Framework}: We show how to identify state average treatment effects from posteriors supplied by an independent latent class model, and formally characterize the asymptotic bias inherent in naive ``hard'' assignment and ``soft'' posterior-weighting estimators. We further characterize what the estimator targets when the identifying assumptions are violated.
    \item \textbf{Orthogonal Estimation}: We adapt the classical two-step orthogonal score adjustment to posterior unmixing. By augmenting the moment equation with a correction term calculated from the HMM's likelihood score, our approach cancels first-order sensitivity to HMM fitting error and supports standard influence-function confidence intervals.
    \item \textbf{Large-Scale Empirical Application}: We apply this methodology to a real-world, large-scale A/B test at Netflix. Our latent-state analysis reveals that heavy-engagement users show significantly positive effects across treatments, with limited evidence of benefit among inactive or rewatch-oriented users.
\end{itemize}

\section{Related Works}
\paragraph{Heterogeneous Treatment Effects.} Estimating heterogeneous treatment effects is an important problem in causal inference.
\citet{athey2016recursive} and \citet{wager2018estimation} use causal trees and forests to learn the conditional average treatment effect (CATE) with respect to observed covariates. Another line of work proposes meta-learners that can use any machine-learning method to estimate the CATE
\citep{kunzel2019metalearners,curth2021nonparametric,nie2021quasi,kennedy2023towards, lan2025meta}. Many works also look at HTEs with respect to learned covariates, such as the propensity score \citep{xie2012estimating}, predicted baseline outcome or risk scores \citep{abadie2018endogenous,chalkou2021two}, and ML predictions of the CATE itself \citep{chernozhukov2025fisher}. We instead target treatment effects within latent classes that are not observed and so cannot be inferred from the outcome, treatment, and covariate data alone.

\paragraph{Latent Variables for Causal Identification.}
A separate literature uses latent-variable models to address unobserved
or imperfectly measured confounding in observational data through
proxy measurement models \citep{kuroki2014measurement},
deep generative models
\citep{louizos2017causal,zhang2021treatment,zhu2025causal},
matrix factorization \citep{kallus2018causal}, or
balancing over latent-confounder distributions
\citep{bennett2019policy,bennett2021offpolicy}.
Other approaches model latent mediators for front-door adjustment
\citep{xu2024causal}, or use factor models for shared confounding across
multiple causes \citep{wang2019blessings,wang2021proxy,miao2023identifying},
including longitudinal treatment assignments \citep{bica2020time}.
Causal identification in these settings requires substantive assumptions
beyond fitting the observed distribution
\citep{damour2019reflections,rissanen2021critical}.
We here assume treatment is randomized, immediately yielding identification of average effects.
We instead use latent variables to analyze effect heterogeneity.

\paragraph{Treatment Effects for Latent Subgroups.}
\emph{One-step approaches} jointly model classes and outcomes or treatment
in observational
\citep{bartolucci2016causal,kim2019estimating,lyu2023estimating} or
experimental \citep{Shahn2017mixture} data, so fitted classes can depend
on the outcome or treatment contrast. \emph{Stepwise approaches} first
fit a measurement model; bias-adjusted three-step methods then correct class misclassification when relating classes to external variables
\citep{bolck2004estimating,vermunt2010latent,bakk2013estimating}.
We similarly separate class learning from effect estimation, but learn
classes from pre-treatment covariates and unmix randomized treatment effects using posteriors latent class probabilities. Our weights are proportional to
Bolck--Croon--Hagenaars (BCH) weights
(for a detailed comparison, see Section~\ref{sec:identification}).

\paragraph{Latent Treatment Effects from Proxy Moments.}
Proxy-moment methods identify class-specific average treatment effects through conditional-independence and rank conditions of observed proxies and do not estimate unit-level class posteriors \citep{mazaheri2025synthetic,virk2026spectral}.
We instead measure effects within separately learned behavioral classes by posterior unmixing under treatment-effect sufficiency and posterior rank conditions (Section~\ref{sec:identification}).


\paragraph{Two-Step Estimation and Neyman Orthogonality.}
Our estimator fits the HMM by maximum likelihood before using its
filtered posteriors in an effect estimating equation. Classical two-step
theory adds a term involving the first-step likelihood score and the
sensitivity of the effect moment to HMM parameters to the influence
function \citep{newey1984method,murphy1985estimation,newey1994large}.
This local-insensitivity construction is related to Neyman's
\(C(\alpha)\) and modern orthogonal estimators
\citep{neyman1959optimal,chernozhukov2018}. We use the same score
adjustment in the target estimating equation
(Section~\ref{sec:orthogonal_correction}). Classical two-step inference leaves the estimate
unchanged and adjusts its variance, while our correction changes the estimate as the HMM is learned on a separate sample. To first order, the correction is equivalent to a one-step Newton update of the frozen HMM on target histories \citep{bickel1975one}, without rerunning the expectation-maximization (EM) algorithm.

\section{Identification of Latent-Class Treatment Effects}
\label{sec:identification}

For experimental units \(i=1,\ldots,n\), let \(H_i\) be the observed pre-treatment covariates, \(Z_i\in\{1,\ldots,K\}\) denote unit $i$'s latent class at assignment, \(A_i\in\{0,1\}\) randomized treatment with known probability \(e\in(0,1)\), and \(Y_i(a)\) its potential outcome.  In our context, we interpret $H_i$ as the history of user $i$ and $Z_i$ as the user's transient but unobserved state; we note that, in general, $H_i$ and $Z_i$ need not have this dynamic interpretation.

The state
posterior is the vector
\(\qb_0(H_i)=(q_{01}(H_i),\ldots,q_{0K}(H_i))^\top\), where
\(q_{0s}(H_i)=\PP(Z_i=s\mid H_i)\). Let \(\pi_s=\PP(Z_i=s)\) denote the prevalence of state $s$.
For states with positive
prevalence, we target the state average treatment effects
\begin{equation*}
  \taub_0=(\tau_{01},\ldots,\tau_{0K})^\top,
  \,
  \tau_{0s}=\E\{Y_i(1)-Y_i(0)\mid Z_i=s\}.
\end{equation*}
To identify the state-specific effects from the identified but potentially complex history-specific CATEs, we make the following sufficiency assumption.

\begin{assumption}[Treatment-effect sufficiency]\label{assum:effect_sufficiency}
For each state \(s\) with positive prevalence,
\(\E[Y_i(1)-Y_i(0)\mid Z_i=s,H_i]=\tau_{0s}\) almost surely
conditional on \(Z_i=s\).
\end{assumption}
This assumption says that, within a state, the remaining history does not modify the conditional average treatment effect.
Iterated expectation then gives the observable mixture
\begin{equation}
  \delta(H_i)=\qb_0(H_i)^\top\taub_0.
  \label{eqn:mixture}
\end{equation}
Let \(Q_0=\E\{\qb_0(H_i)\qb_0(H_i)^\top\}\). For a symmetric
matrix \(Q\), let \(\lambda_{\min}(Q)\) denote its smallest eigenvalues. If
\(\lambda_{\min}(Q_0)>0\), multiplying the mixture identity by
\(\qb_0(H_i)\) and taking expectations identifies the state effects:
\begin{equation*}
  \taub_0=Q_0^{-1}\E[\qb_0(H_i)\delta(H_i)].
\end{equation*}
This rank condition requires variation in the posterior vectors. 

\paragraph{Comparison with Prior Work} Similar to \cite{bolck2004estimating} and \cite{vermunt2010latent}, our approach is also stepwise but motivated by the mixture identity $\delta(H_i)= \qb_0(H_i)^\top\taub_0$, and we unmix the latent class effects by projecting the history-conditional effect onto the posterior vector. Although the motivations differ, the two formulations lead to proportional weights: our unmixing weights satisfy $\{Q_0^{-1} \qb_0(H_i)\}_s = w_s(H_i)/\pi_s$, where $w_s$ is the BCH weight under posterior assignment \citep{bolck2004estimating, vermunt2010latent}. The scale factor reflects different targets. For some target variable $F$ that is mean-independent of $H$ given $Z$, the BCH weights replace the class indicator, i.e. $\E[F\,w_s(H_i)]=\E[F\,\ind(Z_i=s)]$. Our weights target the conditional mean $\E[F\mid Z_i=s]$ directly.

We differ from this literature in three other aspects. First, we formally identify causal treatment effects within latent states. Second, our dynamic latent states are learned from an HMM fitted on pre-treatment histories on a separate sample, which allows for comparison of results across different downstream tasks as the same state definition is shared. Third, rather than only adjusting standard errors for first-stage uncertainty \citep{bakk2014relating}, we remove its first-order effect on the estimate (Section~\ref{sec:orthogonal_correction}).

\paragraph{Within-state Effect Variation.}
The constant-effect restriction in Assumption~\ref{assum:effect_sufficiency} can be relaxed for a prespecified,
low-dimensional pre-treatment summary \(X_i=x(H_i)\):
\begin{equation*}
  \E[Y_i(1)-Y_i(0)\mid Z_i=s,H_i]=\tau_s(X_i).
\end{equation*}
The observable effect then satisfies
\(\delta(H_i)=\sum_{s=1}^Kq_{0s}(H_i)\tau_s(X_i)\).
If \(Q_0(x)=\E[\qb_0(H_i)\qb_0(H_i)^\top\mid X_i=x]\)
is nonsingular for almost every \(x\) in the target population, conditional inversion identifies \((\tau_s(x))_{s=1}^K\).
Pre-treatment baseline activity or unit-level intensity summaries
are possible choices for \(X_i\). This extension requires conditional
posterior rank, which is stronger than the global rank condition above.
The bias and inference results in the next section use the constant-effect case.
Appendix~\ref{app:effect-sufficiency-violation} shows that unmodeled
within-state variation makes the same estimator target a posterior
projection of $\delta(H)$ and bounds its gap from the state ATEs.

\subsection{Bias of Naive Estimators}\label{sec:bias_naive}
For nonnegative weights \(w=(w_1,\ldots,w_n)\), define the weighted
difference-in-means estimator:
\[
\widehat\Delta_n(w)
:=\frac{\sum_i A_iw_iY_i}{\sum_i A_iw_i}
-\frac{\sum_i(1-A_i)w_iY_i}{\sum_i(1-A_i)w_i},
\]
when both denominators are positive. Let \(\widehat q_{ik}\) be the
fitted posterior coordinate and
\(\widehat I_{ik}=\ind\{k=\argmax_{s\le K}\widehat q_{is}\}\),
with a fixed tie rule. Hard assignment uses
\(\widehat\Delta_n(\widehat I_{ik})\); soft weighting uses
\(\widehat\Delta_n(\widehat q_{ik})\).

Under Assumption~\ref{assum:effect_sufficiency}, these contrasts
generally target mixtures even with the true posterior. Write
\(q_{ik}=q_{0k}(H_i)\) and
\(I_{ik}=\ind\{k=\argmax_{s\le K}q_{is}\}\). Define the oracle
versions by replacing \(\widehat q_{ik}\) and \(\widehat I_{ik}\) by
\(q_{ik}\) and \(I_{ik}\), respectively. They satisfy
\begin{align*}
  \widehat\tau_{k,\mathrm{or}}^{\,\mathrm{soft}}
  &\xrightarrow{p}
  \sum_{s=1}^K\frac{\E[q_{ik}q_{is}]}{\E[q_{ik}]}\tau_{0s},
  \\
  \widehat\tau_{k,\mathrm{or}}^{\,\mathrm{hard}}
  &\xrightarrow{p}
  \sum_{s=1}^K\PP(Z_i=s\mid I_{ik}=1)\tau_{0s}.
\end{align*}
Each set of weights sums to one, so the naive estimators recover convex sums of the true latent state aggregates. A derivation of the bias is in Section~\ref{app:naive-bias}.


\section{Estimation}
\label{sec:unmixing-estimation}
In our application, we estimate latent-state effects in a randomized experiment using an HMM
trained on behavioral histories from a separate set of units. Again, we emphasize that the method can also be applied in non-dynamic settings to infer static latent classes. Once
trained, the HMM can be reused across treatments and experiments with the same
pre-treatment history distribution, facilitating comparison. For each experiment, we obtain
posterior state probabilities for its analysis units, estimate the
state effects, and evaluate the orthogonal correction on those units. Our theory assumes that the number of latent states is correctly specified. 

\subsection{The Unmixing Estimator}
To estimate the state effects from \eqref{eqn:mixture}, write
\(\widetilde A_i=A_i-e\) and \(v_A=e(1-e)\).
By consistency, \(Y_i=Y_i(0)+A_i\{Y_i(1)-Y_i(0)\}\).
Randomization then gives
\begin{equation}
  \begin{aligned}
  &\E[\widetilde A_iY_i\mid H_i]\\
    &=\E[\widetilde A_i\mid H_i]\E[Y_i(0)\mid H_i]+\E[\widetilde A_iA_i\mid H_i]\delta(H_i)\\
    &=e(1-e)\delta(H_i)=v_A\delta(H_i)
  \end{aligned}
  \label{eq:unmix-randomization}
\end{equation}
Because the randomization probability \(e\) is known, this moment uses
\(Y_i\) directly; the outcome regression on the history
\(m_0(H)=\E[Y_i\mid H_i=H]\) need not be estimated.
Let \(\eta\) denote the HMM parameter, with true value \(\eta_0\),
and write \(\qb_\eta(H)\) for its filtered state-probability vector
based only on pre-treatment history. After label alignment,
\(\qb_{\eta_0}(H)=\qb_0(H)\).
For a coefficient vector \(\taub\in\R^K\), define
\begin{equation*}
  \bm{\psi}_i(\taub,\eta)
  =\qb_\eta(H_i)\widetilde A_i
   \{Y_i-\widetilde A_i\qb_\eta(H_i)^\top\taub\}.
\end{equation*}
By \eqref{eq:unmix-randomization}, its population mean is
\begin{equation}
  \E\{\bm{\psi}_i(\taub,\eta)\}
  =v_A\E\!\left[\qb_\eta(H_i)
      \{\delta(H_i)-\qb_\eta(H_i)^\top\taub\}\right].
  \label{eq:unmix-population-moment}
\end{equation}
By \eqref{eqn:mixture}, this mean vanishes at
\((\taub_0,\eta_0)\).

Fit \(\widehat\eta\) on the separate HMM training sample. Write \(\widehat{\qb}_i=\qb_{\widehat\eta}(H_i)\).
Let \(\widehat{\taub}^{\,\mathrm{pl}}\) denote the plain plug-in
unmixing estimator, obtained by solving the sample effect equation.
When the sample Gram matrix is invertible, it is
\begin{align}
  \widehat M&=\frac1n\sum_{i=1}^n
     \widetilde A_i^2\widehat{\qb}_i
     \widehat{\qb}_i^\top,\nonumber\\
  \widehat{\taub}^{\,\mathrm{pl}}
  &=\widehat M^{-1}\frac1n\sum_{i=1}^n
     \widehat{\qb}_i\widetilde A_iY_i.
  \label{eq:unmix-plugin}
\end{align}
The population Gram matrix is \(M_0=v_AQ_0\).

\subsection{Orthogonal Correction and Inference}\label{sec:orthogonal_correction}
Estimation of the external HMM can shift the plug-in effect equation
at first order. Let \(p_\eta(H)\) be the HMM likelihood of the
pre-treatment history, with \(\eta\) in identifiable free coordinates,
and let \(D_\eta\) denote differentiation with respect to \(\eta\).
We use the observed-history score
\(\bm{s}^H(H;\eta)=D_\eta\log p_\eta(H)\). Define
\begin{align}
  G_0&=\E[D_\eta\bm{\psi}_i(\taub_0,\eta)|_{\eta_0}],\nonumber\\
  J_0&=\E[D_\eta\bm{s}^H(H_i;\eta_0)],
    \nonumber\\
  B_0&=G_0J_0^{-1}.
  \label{eq:unmix-hmm-sensitivity}
\end{align}
\(G_0\) and \(J_0\) measure how the effect moment and mean history
score change with \(\eta\).
For a conformable matrix \(B\), define the corrected score
\begin{equation*}
  \bm{\phi}_i(\taub,\eta,B)
  =\bm{\psi}_i(\taub,\eta)-B\bm{s}^H(H_i;\eta).
\end{equation*}
At \(\eta_0\), the history score has mean zero, so this subtraction
preserves the target. Its mean is also insensitive to small HMM
parameter errors: under the local conditions below, its derivative
with respect to \(\eta\) at \(\eta_0\) is
\begin{equation*}
  G_0-B_0J_0=\mathbf 0.
\end{equation*}
HMM fitting error therefore enters the mean corrected score
quadratically, together with a product of the errors in
\(\widehat\eta\) and \(\widehat B\).

The expectations defining \(B_0\) are over the target-unit
distribution, so we estimate them on the experiment units. Denote
\(\E_n [f]=n^{-1}\sum_{i=1}^n f(W_i)\). Using
\(\widehat{\taub}^{\,\mathrm{pl}}\) as a preliminary value and holding
it fixed in \(D_\eta\bm{\psi}\), estimate \(B_0\) by
\begin{align}
  \widehat B
  &=\E_n \bigl[D_\eta\bm{\psi}  (\widehat{\taub}^{\,\mathrm{pl}},\widehat\eta)\bigr] \times \E_n[D_\eta\bm{s}^H(H;\widehat\eta)]^{-1}.
\label{eq:unmix-sensitivity-estimator}
\end{align}

The corrected estimator evaluates the HMM score on the target analysis
units:
\begin{equation}
  \widehat{\taub}^{\,\perp}
  =\widehat{\taub}^{\,\mathrm{pl}}
   -\widehat M^{-1}\widehat B
       \E_n [\bm{s}^H(H;\widehat\eta)].
  \label{eq:unmix-orthogonal-estimator}
\end{equation}
At an exact interior likelihood fit, the training sample's mean HMM score is
zero. The target sample's mean score is generally nonzero and cancels
the first-order effect of estimating \(\eta_0\).
For each target unit, define
\(\widehat{\bm{\phi}}_i
  =\bm{\phi}_i(\widehat{\taub}^{\,\perp},
               \widehat\eta,\widehat B)\).
The sandwich covariance estimator is
\begin{align}
  \widehat\Omega
  &=\frac1n\sum_i
     \widehat{\bm{\phi}}_i\widehat{\bm{\phi}}_i^\top,
  \nonumber\\
  \widehat{\operatorname{Var}}
       (\widehat{\taub}^{\,\perp})
  &=\frac1n\widehat M^{-1}\widehat\Omega\widehat M^{-1}.
  \label{eq:unmix-sandwich}
\end{align}

\begin{theorem}[Orthogonal inference with an external HMM]
\label{thm:orthogonal-unmixing}
Let the analysis units be i.i.d., independent of i.i.d. HMM
training histories drawn from the same distribution as \(H_i\). Suppose
\(A_i\sim\operatorname{Bernoulli}(e)\), independently of
\((H_i,Y_i(0),Y_i(1))\), with known \(e\in(0,1)\). Assume observed
outcomes are consistent and
\(\E\{Y_i(a)^2\}<\infty\) for \(a\in\{0,1\}\). Assume the HMM is
correctly specified with fixed \(K\) and a fixed-dimensional interior
parameter \(\eta_0\), after a label alignment based only on its
training histories. Suppose that Assumption ~\ref{assum:effect_sufficiency} holds,
\(\lambda_{\min}(Q_0)\geq\kappa_0>0\),
\(\E\{\bm{s}^H(H_i;\eta_0)\}=\mathbf 0\),
\(\E\big[\norm{\bm{s}^H(H_i;\eta_0)}^2\big]<\infty\), and nonsingular \(J_0\).

Assume the local \(L_2\) continuity and
second-order mean expansions in
Assumption~\ref{ass:unmix-local-regularity}. Suppose the externally
fitted \(\widehat\eta\) and the target-sample \(\widehat B\) satisfy
\begin{equation*}
  \begin{aligned}
  \norm{\widehat\eta-\eta_0}&=O_p(r_\eta),\\
  \opnorm{\widehat B-B_0}&=O_p(r_B),
  \end{aligned}
\end{equation*}
for rates \(r_\eta,r_B\to0\).
Denote \(\bm{\phi}_{0i}
  =\bm{\phi}_i(\taub_0,\eta_0,B_0)\).
Then HMM fitting error contributes no linear term:
\begin{equation}
  \begin{aligned}
  \widehat{\taub}^{\,\perp}-\taub_0
  &=M_0^{-1}\frac1n\sum_{i=1}^n\bm{\phi}_{0i}\\
  &\quad+O_p(r_\eta^2+r_\eta r_B)+o_p(n^{-1/2}).
  \end{aligned}
  \label{eq:unmix-orthogonal-expansion}
\end{equation}
Consequently,
\begin{equation}
  \norm{\widehat{\taub}^{\,\perp}-\taub_0}
  =O_p\!\left(
    \frac{n^{-1/2}+r_\eta^2+r_\eta r_B}{v_A\kappa_0}
    \right).
  \label{eq:unmix-rate}
\end{equation}
If \(\sqrt n(r_\eta^2+r_\eta r_B)\to0\), put
\(\Omega_0=\E\{\bm{\phi}_{0i}\bm{\phi}_{0i}^\top\}\). Then
\begin{equation}
  \sqrt n(\widehat{\taub}^{\,\perp}-\taub_0)
  \xrightarrow{d}\mathcal N
      (\mathbf 0,M_0^{-1}\Omega_0M_0^{-1}).
  \label{eq:unmix-clt}
\end{equation}
Under the same condition, \(n\) times the sandwich covariance in
\eqref{eq:unmix-sandwich} converges in probability to the displayed
asymptotic covariance.
\end{theorem}
Appendix~\ref{app:orthogonal-unmixing-proof} first verifies the
target-sample correction rate in
Lemma~\ref{lem:unmix-sensitivity-estimation}, then proves the theorem.

For each state, an approximate 95\% confidence interval is the
corrected effect estimate plus or minus 1.96 times its sandwich
standard error from \eqref{eq:unmix-sandwich}. This normal
approximation requires a positive limiting variance. Weak posterior
rank can inflate the standard error through \(M_0^{-1}\).

\section{Simulation experiments}
\label{sec:simulations}

We use two designs to separate ambiguity about a latent state from error
in estimating the state model. In the oracle-posterior design, every
method receives the exact posterior \(\qb_0(H_i)\), which isolates the
bias of hard assignment and posterior weighting. In the full workflow,
we fit an HMM on histories from separate units and compare plain
unmixing with the orthogonal correction on a randomized target sample.
All experiments use known treatment probability \(e=1/2\);
the centered outcome model has \(\E(Y_i\mid H_i)=0\).
Appendix~\ref{app:simulations} gives the designs and additional results.

\subsection{Oracle Posteriors: Bias and Identification}

We first use three equally prevalent states with
\(\taub_0=(-1,0,1)^\top\). A categorical pre-treatment signal
\(H_i\in\{1,2,3\}\) equals the state \(Z_i\) with probability
\(\rho=\PP(H_i=Z_i)\in\{0.40,0.60,0.80,0.95,1.00\}\); under the
balanced prior, \(\rho\) is hard-assignment accuracy. With independent
standard normal noise \(\veps_i\), outcomes follow
\(Y_i=(A_i-1/2)\tau_{0,Z_i}+\veps_i\). We use
\(n\in\{500,2{,}000,10{,}000,50{,}000\}\) units and 1,000
replications per condition. No HMM is estimated in this design.

Even with the exact posterior, hard and soft estimators target mixtures
of the state effects. Table~\ref{tab:overlap-results} presents the
root mean squared error (RMSE) of the three estimators as signal
accuracy varies at \(n=50{,}000\). All methods become less precise as
accuracy falls, but unmixing has lower RMSE than hard or soft weighting
in this design. Its intervals have near-nominal coverage of the true
state effects, whereas the naive intervals often undercover. When
\(\rho=1\), the posterior probability exactly identifies the latent state, and the three methods agree up to finite-sample differences. Appendix
Tables~\ref{tab:overlap-rmse-full} and
\ref{tab:overlap-coverage-full} report all sample sizes.

\begin{table}[t]
\centering
\caption{\small Posterior-overlap results at \(n=50{,}000\), based on 1,000
replications. Signal accuracy is \(\rho=\PP(H_i=Z_i)\). Entries are
estimates (Monte Carlo standard errors).
RMSE weights the three states equally; coverage averages across states
for the true latent-state effects. Posterior-based methods use the
exact Bayes posterior.}
\label{tab:overlap-results}
\scriptsize
\setlength{\tabcolsep}{2pt}
\begin{tabular}{@{}lccc@{}}
\toprule
\multicolumn{4}{l}{\textbf{RMSE \(\downarrow\)}} \\
\cmidrule(lr){1-4}
\(\rho\) & Unmixing & Hard & Soft \\
\midrule
0.40 & 0.1393 (0.0021) & 0.7349 (0.00032) & 0.8084 (0.000032) \\
0.60 & 0.0347 (0.00051) & 0.4897 (0.00029) & 0.6858 (0.00012) \\
0.80 & 0.0214 (0.00028) & 0.2453 (0.00030) & 0.4165 (0.00021) \\
0.95 & 0.0160 (0.00022) & 0.0633 (0.00027) & 0.1189 (0.00025) \\
1.00 & 0.0154 (0.00019) & 0.0154 (0.00019) & 0.0154 (0.00019) \\
\addlinespace[6pt]
\multicolumn{4}{l}{\textbf{True-effect 95\% coverage}} \\
\cmidrule(lr){1-4}
\(\rho\) & Unmixing & Hard & Soft \\
\midrule
0.40 & 0.951 (0.0044) & 0.317 (0.0023) & 0.318 (0.0022) \\
0.60 & 0.953 (0.0042) & 0.317 (0.0023) & 0.316 (0.0023) \\
0.80 & 0.948 (0.0043) & 0.318 (0.0022) & 0.319 (0.0022) \\
0.95 & 0.959 (0.0036) & 0.315 (0.0025) & 0.314 (0.0025) \\
1.00 & 0.948 (0.0040) & 0.948 (0.0040) & 0.948 (0.0040) \\
\bottomrule
\end{tabular}
\end{table}

State imbalance also weakens identification:
at \(\rho=0.80\), reducing the smallest state prevalence from \(1/3\)
to 0.10 raises the Gram condition number from 2.04 to 18.12 and unmixing's equal-state RMSE from 0.021 to 0.037. Appendix~\ref{app:sim-imbalance} reports the imbalance results; the unmixing estimator consistently outperforms the naive estimators. 

\subsection{Estimated HMM: Orthogonal Correction}

We next generate 21-day action and viewing-intensity histories from a
correctly specified five-state HMM and fit it by EM on independent panels
from the same history distribution, with
\(N\in\{300,600,1{,}200,2{,}400\}\) units. In each of 200 paired
replications, a randomized 50,000-unit target panel is analyzed at
sample-sizes of \(n\in\{2{,}000,10{,}000,50{,}000\}\).
The state effects, ordered by activity, are
\(\taub_0=(-1.2,-0.6,0,0.6,1.2)^\top\).

The HMM fit uses a separate set of units with behavioral panel, and the correction's posterior derivatives, score information, and effect moments are calculated on the target experiment sample. In the implementation of the orthogonal correction term, the target score Jacobian is approximated by the outer product of observed-history scores (OPG). Under correct HMM specification, the information-matrix identity equates the target expected score Jacobian to the negative expected score outer product, motivating the OPG approximation \citep{berndt1974estimation} without repeated numerical differentiation. Appendix~\ref{app:sim-hmm} describes the HMM, orthogonal correction term, and additional effect profiles.

Figure~\ref{fig:orthogonal-opg-grid} shows that the RMSE gain grows
with target size \(n\) at every HMM training size \(N\), and generally
shrinks as \(N\) grows. In the primary comparison pooling \(N=600\)
and \(N=1{,}200\) at \(n=50{,}000\), orthogonal correction reduces
equal-state mean squared error (MSE), averaged over replications, by
31.4\%. At \(N=600\), empirical
interval coverage rises from 85.9\% to 93.0\%. The sample-size pattern
is consistent with correction attenuating HMM fitting error: with
fewer training histories there is more such error to remove, while a
larger target sample makes it more visible relative to sampling noise.

Theorem~\ref{thm:orthogonal-unmixing} helps explain this pattern. Let
\(r_\eta\) denote the HMM estimation rate, which depends on the number
\(N\) of training histories. Plain unmixing has error
\(O_p(n^{-1/2}+r_\eta)\). Under the conditions of Lemma~\ref{lem:unmix-sensitivity-estimation},
orthogonal unmixing reduces that contribution to second order, giving
\(O_p(n^{-1/2}+r_\eta^2)\). Correction can therefore help more when the
HMM is trained on fewer histories. As the target sample grows, sampling
noise falls and the correction coefficient is estimated more accurately, which makes the gain easier to detect. Appendix~\ref{app:sim-hmm-results}
reports the complete $(n,N)$ and effect-profile results.

\begin{figure}
\centering
\includegraphics[width=0.6\columnwidth]{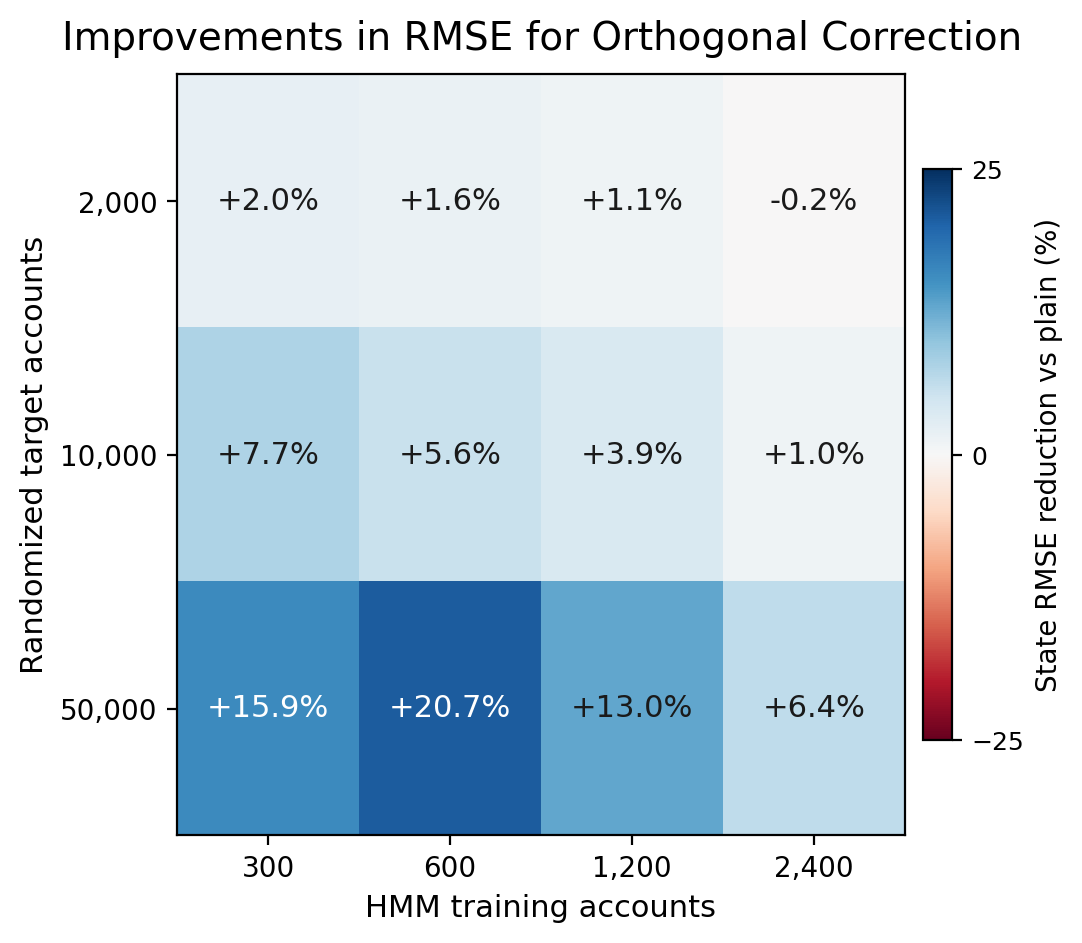}
\caption{\small Equal-state RMSE reduction from target-sample OPG
correction relative to plain unmixing under monotone effects. Each cell
uses 200 paired replications; positive values favor correction.}
\label{fig:orthogonal-opg-grid}
\end{figure}

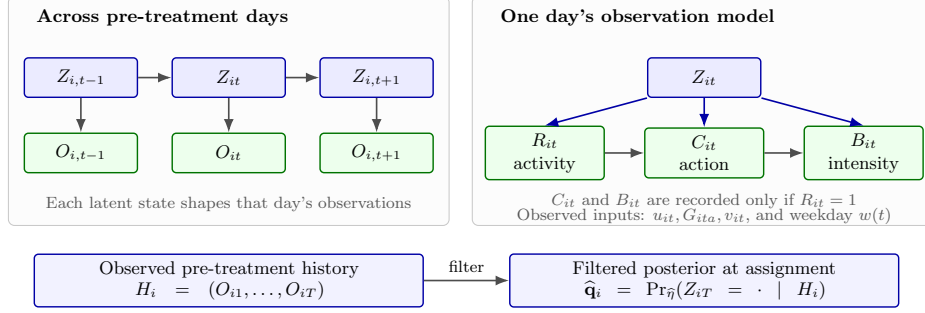
\begin{figure*}[t]
\centering
\resizebox{0.75\textwidth}{!}{
\begin{tikzpicture}[
  x=1cm,y=1cm,>=Latex,
  every node/.style={font=\footnotesize},
  state/.style={draw=blue!65!black,fill=blue!8,rounded corners=2pt,
    minimum width=1.95cm,minimum height=0.68cm,align=center,line width=0.7pt},
  observed/.style={draw=green!45!black,fill=green!7,rounded corners=2pt,
    minimum width=1.95cm,minimum height=0.68cm,align=center,line width=0.7pt},
  layer/.style={draw=green!45!black,fill=green!7,rounded corners=2pt,
    minimum width=2.05cm,minimum height=0.82cm,align=center,line width=0.7pt},
  output/.style={draw=blue!65!black,fill=blue!5,rounded corners=2pt,
    text width=6.45cm,minimum height=0.88cm,align=center,line width=0.7pt},
  link/.style={->,draw=black!70,line width=0.8pt}
]
  \draw[draw=black!22,fill=black!1,rounded corners=3pt]
    (0,0.8) rectangle (7.6,4.75);
  \node[anchor=west,font=\footnotesize\bfseries] at (0.35,4.38)
    {Across pre-treatment days};
  \node[state] (zprev) at (1.25,3.35) {$Z_{i,t-1}$};
  \node[state] (zmid)  at (3.80,3.35) {$Z_{it}$};
  \node[state] (znext) at (6.35,3.35) {$Z_{i,t+1}$};
  \draw[link] (zprev.east) -- (zmid.west);
  \draw[link] (zmid.east) -- (znext.west);
  \node[observed] (oprev) at (1.25,2.05) {$O_{i,t-1}$};
  \node[observed] (omid)  at (3.80,2.05) {$O_{it}$};
  \node[observed] (onext) at (6.35,2.05) {$O_{i,t+1}$};
  \draw[link] (zprev.south) -- (oprev.north);
  \draw[link] (zmid.south) -- (omid.north);
  \draw[link] (znext.south) -- (onext.north);
  \node[font=\scriptsize,text=black!65] at (3.80,1.17)
    {Each latent state shapes that day's observations};

  \draw[draw=black!22,fill=black!1,rounded corners=3pt]
    (8.0,0.8) rectangle (16.0,4.75);
  \node[anchor=west,font=\footnotesize\bfseries] at (8.35,4.38)
    {One day's observation model};
  \node[state] (zday) at (12.0,3.35) {$Z_{it}$};
  \node[layer] (activity) at (9.25,2.05) {$R_{it}$\\activity};
  \node[layer] (action) at (12.0,2.05) {$C_{it}$\\action};
  \node[layer] (intensity) at (14.75,2.05) {$B_{it}$\\intensity};
  \draw[link,draw=blue!65!black] (zday.south west) -- (activity.north);
  \draw[link,draw=blue!65!black] (zday.south) -- (action.north);
  \draw[link,draw=blue!65!black] (zday.south east) -- (intensity.north);
  \draw[link] (activity.east) -- (action.west);
  \draw[link] (action.east) -- (intensity.west);
  \node[font=\scriptsize,text=black!65] at (12.0,1.23)
    {$C_{it}$ and $B_{it}$ are recorded only if $R_{it}=1$};
  \node[font=\scriptsize,text=black!65] at (12.0,0.97)
    {Observed inputs: $u_{it},G_{ita},v_{it}$, and weekday $w(t)$};

  \node[output] (history) at (3.80,-0.15)
    {Observed pre-treatment history\\[-1pt]
     $H_i=(O_{i1},\ldots,O_{iT})$};
  \node[output] (posterior) at (12.0,-0.15)
    {Filtered posterior at assignment\\[-1pt]
     $\widehat{\qb}_i=\Pr_{\widehat\eta}(Z_{iT}=\cdot\mid H_i)$};
  \draw[link] (history.east) --
    node[midway,above,font=\scriptsize] {filter} (posterior.west);
\end{tikzpicture}}
\caption{\small Schematic of the user behavior HMM over \(T\)
pre-treatment days. \(Z_{it}\) is user \(i\)'s latent state and
\(O_{it}\) is the observed record on day \(t\). Activity \(R_{it}\)
is followed by action \(C_{it}\) when active and viewing intensity
\(B_{it}\) conditional on action. The inputs \(u_{it},G_{ita},v_{it}\)
are user summaries, and \(w(t)\) is weekday.}
\label{fig:netflix-hmm-schematic}
\end{figure*}

\section{Real World Case Study}
To demonstrate the practical utility of our latent-class HTE estimator, we apply our methodology to evaluate treatment effect heterogeneity due to dynamic latent states in a large-scale randomized A/B test conducted at Netflix.

Personalized recommendation systems are core to improving user engagement at Netflix. However, the impact of these recommendations is plausibly heterogeneous and modified by transient user moods. For example, a user who has recently finished a TV series may be highly receptive to recommendations of novel titles, whereas a user actively progressing through a multi-season show or repeatedly rewatching familiar content is likely to continue their current viewing pattern regardless of recommendations. Because these underlying moods are not directly observed, we characterize them in the above framework as transient latent states that generate different conditional probability distributions over a user’s subsequent actions.

\subsection{Behavioral-state Hidden Markov Model}

Figure~\ref{fig:netflix-hmm-schematic} illustrates the user behavior HMM.
Let \(T\) be the number of pre-treatment days and \(O_{it}\) the observed
record on day \(t\), so that \(H_i=(O_{i1},\ldots,O_{iT})\).
Let \(Z_{it}\) be user \(i\)'s state on pre-treatment day \(t\),
\(R_{it}\) indicate activity, \(C_{it}\) denote the action when active,
and \(B_{it}\) denote a viewing-duration bin. Write \(w(t)\) for the
weekday; \(s,r\) index states, \(a\) indexes actions, and \(b\)
indexes duration bins. The transition matrix is \(A=(A_{sr})\).
Conditional on activity, actions include start, continue, complete,
mixed start-and-complete, and rewatch. The emission model uses fixed
user offsets \(u_{it},G_{ita},v_{it}\), state terms
\(\theta_s,\phi_{sa},\kappa_s\), weekday terms
\(d_{w(t)},\delta_{w(t),a},\psi_{w(t),a,b}\), coefficients
\(\lambda,\beta\), and a numerical score \(z_b\) for duration bin
\(b\):

{\small
\begin{align*}
\PP(Z_{i,t+1}=r\mid Z_{it}=s)&=A_{sr}, \\
\operatorname{logit}\big(\Pr(R_{it}=1\mid Z_{it}=s) \big)&=u_{it}+\theta_s+d_{w(t)}, \\
\PP(C_{it}=a\mid R_{it}=1,Z_{it}=s)&\propto\phi_{sa}e^{\lambda G_{ita}+\delta_{w(t),a}}, \\
\PP(B_{it}=b\mid R_{it}=1,C_{it}=a,Z_{it}=s) &\propto\psi_{w(t),a,b}e^{z_b(\beta v_{it}+\kappa_s)}.
\end{align*}
}

The offsets summarize trailing viewing intensity and action mix, so
state reflects behavior relative to a user's own baseline. We estimate
the HMM's initial-state distribution \(\bm{\nu}\), transition matrix
\(A\), and emission parameters by Baum--Welch EM on a disjoint fold from the treatment effect estimation.

\begin{figure*}[t]
\centering
\includegraphics[width=0.85\textwidth]{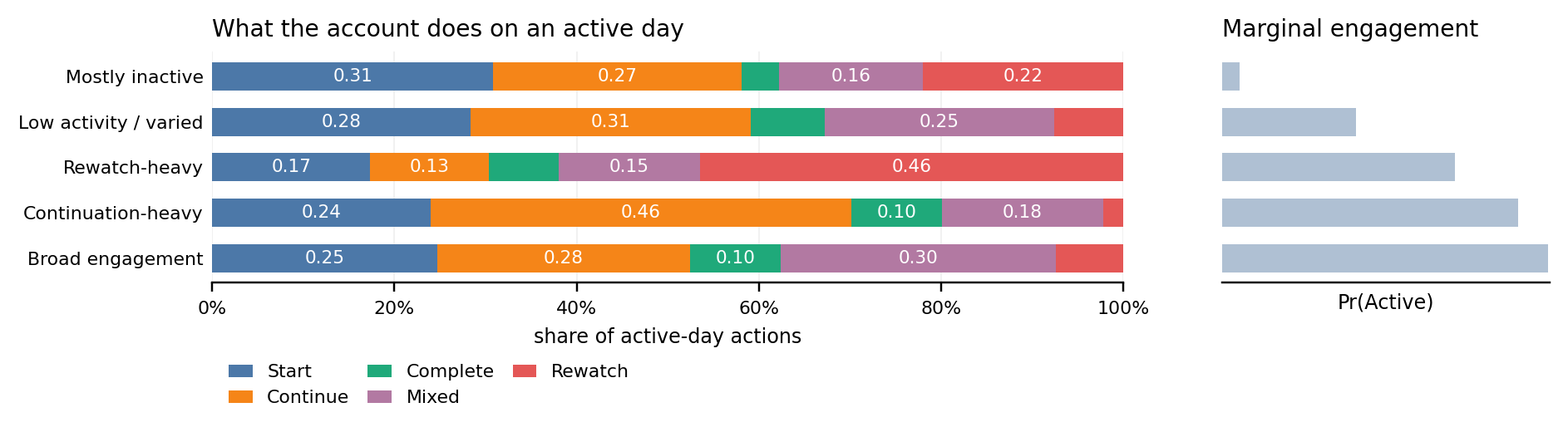}
\caption{\small Behavioral profiles of the fitted five-state HMM. The left panel shows action shares conditional on activity; the right panel shows marginal activity probabilities. The states describe temporary behavioral contexts, not permanent member types.}
\label{fig:state_action}
\end{figure*}

Figure~\ref{fig:state_action} shows the fitted states' activity and
action profiles. Their viewing-intensity distributions conditional on
playback are in Appendix Figure~\ref{fig:state-intensity}.

We choose the number of states using an elbow in out-of-time held-out
log-likelihood per observed day for candidate HMMs with
\(K=2,\ldots,7\) on a smaller sample for hyperparameter selection (Appendix Figure~\ref{fig:hmm-k-selection}). The gain from
4 to 5 states is \(0.00721\), whereas the gain from 5 to 6
falls to \(0.00178\), and adding a $7^{th}$ state gives a similarly small
gain of \(0.00197\). We therefore use \(K=5\).
Appendix Figure~\ref{fig:hmm-state-stability} compares state posteriors
across matched \(K=4,5,6\) states.

\subsection{Latent State HTEs}

We apply the proposed methodology to a real-world A/B test of a new recommendation model architecture with one control arm and three treatment arms (``cells'') and approximately 24 million users. After selecting \(K=5\), we fit the HMM using 180 pre-treatment days on 10\% of the units in the experiment. The fitted HMM is then frozen, and state-specific treatment effects are estimated using the proposed estimator on the remaining 90\% of users. By leveraging the inferred latent states at the time of treatment assignment as effect modifiers, we apply our proposed framework to investigate whether the new recommendation architecture exerts uniform benefits across the user population or if its efficacy is statistically isolated to specific, transient behavioral profiles. Let \(\widehat Q=n^{-1}\sum_{i=1}^n\widehat{\qb}_i\widehat{\qb}_i^\top\) be the posterior Gram matrix on the effect-estimation users. Posterior variation is adequate for inversion, with the minimum eigenvalue $\lambda_{\min}(\widehat Q)=0.041$ and condition number $\operatorname{cond}(\widehat Q)=6.5$.

\begin{figure*}[t]
\centering
\includegraphics[width=0.85\textwidth]{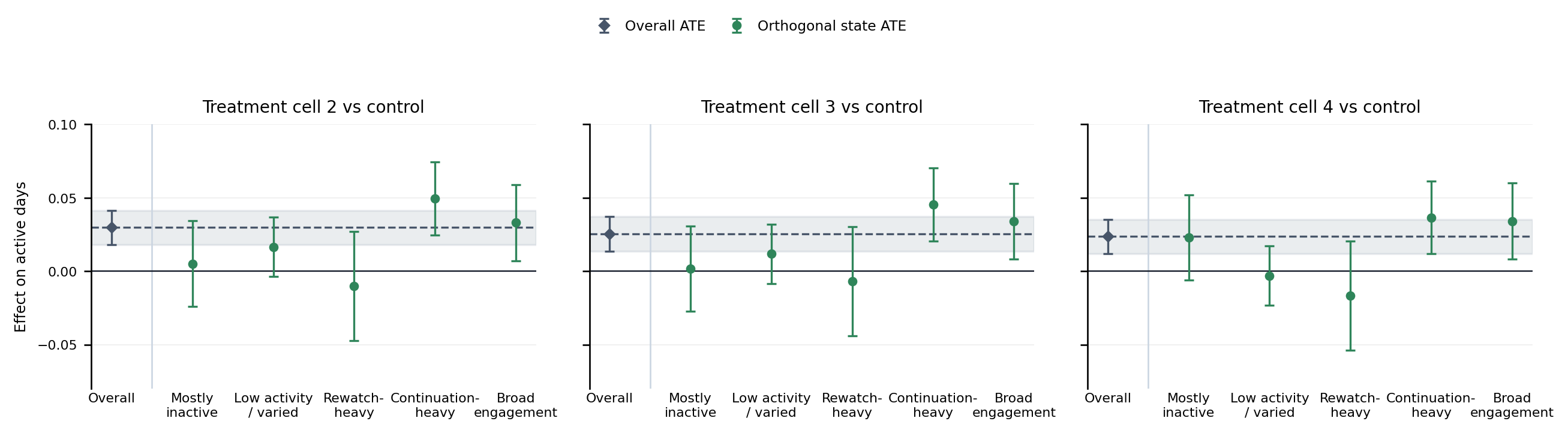}
\caption{\small  Orthogonal latent-state ATEs for a metric that captures active days. Error bars are pointwise 95\% confidence intervals; gray diamonds and bands show the overall ATE. State labels match Figure~\ref{fig:state_action}.}
\label{fig:hte}
\end{figure*}
The learned states represent distinct, interpretable user behaviors as seen in Figure~\ref{fig:state_action}. As Figure~\ref{fig:hte} illustrates, treatment effects vary with the level of engagement (e.g., high vs. low). Estimated latent-state ATEs are not statistically significant in all treatment arms for mostly inactive, low-activity, and rewatch-heavy users; these moods also tend to have the lowest marginal engagement. By contrast, continuation-heavy and broad-engagement users, who are highly engaged, show significantly positive effects across all three treatments. Taken together, the treatments appear most effective among users exhibiting ongoing or broad engagement, with limited evidence of benefit among inactive or rewatch-oriented users.

\section{Discussion and Limitations}
\label{sec:discussion}

In this work, we demonstrated that an HMM trained on pre-treatment behavioral histories can be reused to recover heterogeneous treatment effects across transient latent states. Our unmixing estimator successfully isolates true state-specific effects even when individual state assignments remain uncertain. Furthermore, our orthogonal correction mitigates sensitivity to HMM estimation error, yielding robust state-level inference at scale.

\paragraph{Limitations and Assumptions.}
Interpreting the estimated coefficients as state ATEs requires treatment-effect sufficiency and the number of states ($K$) to be correctly specified. In general, posterior unmixing estimates the best linear approximation of the CATE conditioning on the history (i.e. $\delta(H)$) using the fitted posteriors. With a misspecified state count, the estimated effects will be a weighted sum of the true state ATEs (Appendix~\ref{app:wrong-K}). Under correct HMM specification, the expected score Jacobian equals the negative expected score outer product (OPG approximation), but this equality can fail under misspecification \citep{white1982maximum}. In our application, we performed a hyperparameter tuning to select the number of latent states.

\paragraph{Computational Considerations.}
The orthogonal correction requires estimating and inverting the HMM score Jacobian. We use the OPG approximation of the Jacobian to avoid the prohibitive cost of repeated score evaluations required by finite differences. Because a nearly singular Jacobian can amplify error, in cases of extreme ill-conditioning, aggressive ridge regularization or plain unmixing may be preferable.

\paragraph{Future Work.} A natural extension of this framework is estimating latent-state effects in observational settings. This could involve unmixing a cross-fitted doubly robust pseudo-outcome \citep{kennedy2023towards} for the CATE conditional on the behavioral history. The challenge arises when the latent state $Z_i$ affects both treatment and potential outcomes beyond $H_i$. The observed propensity is a mixture of state-specific propensities, so treated and untreated accounts with the same history can have different state mixtures. A pseudo-outcome based only on $H_i$ need not represent a causal effect, and unmixing it would require additional assumptions about treatment assignment and outcomes within states.

\subsection*{AI use statement}

In this work, we used generative AI tools for: generate synthetic data sets (as part of the simulation experiments), formulate mathematical claims, assist in the writing of proofs, design or provide feedback on experiments, implement methods, interpret results (simulation results only).

We have not used generative AI tools for: help develop theoretical models or conceptual frameworks, provide critical ingredients for proving mathematical claims (proof sketch were provided to the AI models to help write the proof), and propose or refine hypotheses, assist with translation, clean and reformat datasets, support qualitative and thematic data analysis are not applicable to this work.

Additionally, we used generative AI tools for tasks: create or modify scientific figures or images, suggest experimental parameters, create or edit software code, draft parts of a research paper, edit a research paper to improve readability, identify relevant literature, format references.

We have reviewed all AI-assisted work. We went through all additional related literature/citations that was identified by AI models as well as the code that was produced by LLM to implement the method. All LLM-assisted proofs were checked line by line. 

We take responsibility for the final content of this work, including text, claims, or artifacts produced with the aid of generative AI.

\bibliographystyle{plainnat}
\bibliography{ref}

\clearpage
\appendix

\section{Bias of hard and soft state estimators}
\label{app:naive-bias}

\subsection{Posterior weighting and hard assignment}
\label{app:naive-bias-proof}

Hard assignment treats the most likely state as observed, while soft
weighting treats posterior probability as fractional membership. The
following proposition shows that both contrasts generally mix state
effects, even with the true posterior, and gives the conditions under
which they are biased for a state's effect.

\begin{proposition}[Mixture limits of hard and soft contrasts]
\label{prop:naive-estimator-bias}
Fix a state \(k\) and \(K\geq2\). Suppose units are i.i.d.,
\(A_i\sim\operatorname{Bernoulli}(e)\) independently of
\((H_i,Y_i(0),Y_i(1))\) for \(e\in(0,1)\), observed outcomes are
consistent, and \(\E[|Y_i(a)|]<\infty\) for \(a\in\{0,1\}\).
Write \(q_{is}=\Pr(Z_i=s\mid H_i)\) and suppose Assumption~\ref{assum:effect_sufficiency} holds.
Let \(I_{ik}=\ind\{k=\argmax_{s\leq K}q_{is}\}\), using a fixed tie
rule, and suppose \(\Pr(I_{ik}=1)>0\).
For nonnegative weights \(w=(w_1,\ldots,w_n)\), define, whenever both
denominators are positive,
\begin{equation*}
  \begin{aligned}
  \widehat\Delta_n(w)
  &=\frac{\sum_i A_iw_iY_i}{\sum_i A_iw_i}-\frac{\sum_i(1-A_i)w_iY_i}
               {\sum_i(1-A_i)w_i}.
  \end{aligned}
\end{equation*}
Define their population mixing weights by
\begin{align*}
  \omega_{ks}^{\,\mathrm{soft}}
  &=\frac{\E[q_{ik}q_{is}]}{\E[q_{ik}]},\\
  \omega_{ks}^{\,\mathrm{hard}}
  &=\Pr(Z_i=s\mid I_{ik}=1).
\end{align*}
Define the oracle estimators
\(\widehat\tau_{k,\mathrm{or}}^{\,\mathrm{soft}}
  =\widehat\Delta_n((q_{ik})_i)\) and
\(\widehat\tau_{k,\mathrm{or}}^{\,\mathrm{hard}}
  =\widehat\Delta_n((I_{ik})_i)\). Then
\begin{align*}
  \widehat\tau_{k,\mathrm{or}}^{\,\mathrm{soft}}
  &\xrightarrow{p}T_k^{\,\mathrm{soft}}
   :=\sum_{s=1}^K\omega_{ks}^{\,\mathrm{soft}}\tau_{0s},\\
  \widehat\tau_{k,\mathrm{or}}^{\,\mathrm{hard}}
  &\xrightarrow{p}T_k^{\,\mathrm{hard}}
   :=\sum_{s=1}^K\omega_{ks}^{\,\mathrm{hard}}\tau_{0s}.
\end{align*}
Both sets of weights are nonnegative and sum to one. Consequently, the
probability-limit biases relative to the state-\(k\) effect are
\begin{align}
  T_k^{\,\mathrm{soft}}-\tau_{0k}
  &=\sum_{s\ne k}\omega_{ks}^{\,\mathrm{soft}}
       (\tau_{0s}-\tau_{0k}),\nonumber\\
  T_k^{\,\mathrm{hard}}-\tau_{0k}
  &=\sum_{s\ne k}\omega_{ks}^{\,\mathrm{hard}}
       (\tau_{0s}-\tau_{0k}).
  \label{eq:app-naive-bias}
\end{align}
For either oracle estimator, the bias is nonzero exactly when its
corresponding sum in \eqref{eq:app-naive-bias} is nonzero. In
particular, if \(\Pr(I_{ik}=1,\ q_{ik}<1)>0\), both estimators are
asymptotically biased whenever \(\tau_{0k}\) is a strict maximum or
minimum among states of positive probability. Both limits lie below
\(\tau_{0k}\) at a strict maximum and above it at a strict minimum.
\end{proposition}


\begin{proof}[Proof of Proposition~\ref{prop:naive-estimator-bias}]
First use an oracle weight \(w_i=w(H_i)\in[0,1]\) with
\(\E[w(H_i)]>0\). Put \(p_1=e\) and \(p_0=1-e\).
Randomization, consistency, and the law of large
numbers give, for \(a\in\{0,1\}\),
\begin{align*}
  \frac1n\sum_i\ind\{A_i=a\}w(H_i)Y_i
  &\xrightarrow{p}p_a\E[w(H_i)Y_i(a)],\\
  \frac1n\sum_i\ind\{A_i=a\}w(H_i)
  &\xrightarrow{p}p_a\E[w(H_i)],
\end{align*}
The first-moment assumption makes
these limits finite, and both denominator limits are positive. The
continuous mapping theorem and conditioning on \(H_i\) yield the
effect averaged over the histories selected by the weight:
\begin{equation}
  \widehat\Delta_n((w(H_i))_i)
  \xrightarrow{p}
  \frac{\E[w(H_i)\delta(H_i)]}{\E[w(H_i)]}.
  \label{eq:app-naive-weighted-target}
\end{equation}

On \(\{I_{ik}=1\}\), \(q_{ik}\geq1/K\); thus
\(\E[q_{ik}]\geq\Pr(I_{ik}=1)/K>0\).
For the soft estimator, substitute \(w(H_i)=q_{ik}\) into
\eqref{eq:app-naive-weighted-target} and then use
\eqref{eqn:mixture}:
\begin{equation*}
  \E[q_{ik}\delta(H_i)]
  =\sum_{s=1}^K\E[q_{ik}q_{is}]\tau_{0s}.
\end{equation*}
This establishes the claimed soft-weighting limit. The weights sum to one
because the posterior coordinates do:
\(\sum_s\E[q_{ik}q_{is}]=\E[q_{ik}]\).

For hard assignment, \(I_{ik}\) is a function of \(H_i\), so
\begin{align*}
  \E[I_{ik}\delta(H_i)]
  &=\sum_{s=1}^K\tau_{0s}\E[I_{ik}q_{is}],\\
  \E[I_{ik}q_{is}]
  &=\E\!\left[
      I_{ik}\E[\ind(Z_i=s)\mid H_i]\right]\\
  &=\Pr(I_{ik}=1,Z_i=s).
\end{align*}
Dividing by \(\E[I_{ik}]=\Pr(I_{ik}=1)\) gives the claimed
hard-assignment limit. The hard weights
sum to one because they are conditional state probabilities. Subtracting
\(\tau_{0k}\) from each convex combination proves
\eqref{eq:app-naive-bias}.

Finally, the total weight assigned to states other than \(k\) is
\begin{equation*}
  \begin{aligned}
  \sum_{s\ne k}\omega_{ks}^{\,\mathrm{soft}}
  &=\frac{\E[q_{ik}(1-q_{ik})]}{\E[q_{ik}]},\\
  \sum_{s\ne k}\omega_{ks}^{\,\mathrm{hard}}
  &=\Pr(Z_i\ne k\mid I_{ik}=1).
  \end{aligned}
\end{equation*}
On \(\{I_{ik}=1,\ q_{ik}<1\}\), the maximum-posterior rule implies
\(q_{ik}\geq 1/K>0\), so both quantities are strictly positive when
this event has positive probability. A state of zero probability has
zero mixture weight. Thus, if \(\tau_{0k}\) is a strict maximum, every
nonzero-weight term in \eqref{eq:app-naive-bias} is negative; at a
strict minimum, every nonzero-weight term is positive. This proves
the claimed directions of both biases.
\end{proof}

\section{Orthogonal inference for the state ATEs}
\label{app:orthogonal-unmixing-proof}

Throughout this section, \(\E[f(W)]\) averages over an
independent target unit. If \(f\) depends on the externally fitted HMM,
this expectation and its \(L_2\) norm condition on the HMM training
sample. For a scalar- or vector-valued function \(f\) of an analysis unit,
write \(\norm{f}_{L_2}^2=\E[\norm{f(W_i)}^2]\).

This section separates the correction-rate calculation from the
inference argument. Lemma~\ref{lem:unmix-sensitivity-estimation}
verifies the theorem's assumed rate for the correction computed on
the target units. The subsequent theorem proof establishes
orthogonality, the asymptotic expansion, and consistency of the
sandwich covariance under its stated rate conditions.

\begin{assumption}[Local regularity for orthogonal inference]
\label{ass:unmix-local-regularity}
As \(\eta\to\eta_0\),
\(\qb_\eta(H)\), \(\bm{s}^H(H;\eta)\), and
\(\bm{\psi}(\taub_0,\eta)\) are continuous in \(L_2\).
For \(u=\eta-\eta_0\), their relevant population means have the
second-order expansions
\begin{equation}
  \begin{aligned}
  \E[\bm{\psi}_i(\taub_0,\eta)]
    &=G_0u+O(\norm{u}^2),\\
  \E[\bm{s}^H(H_i;\eta)]
    &=J_0u+O(\norm{u}^2).
  \end{aligned}
  \label{eq:unmix-local-smoothness}
\end{equation}
\end{assumption}

\subsection{Rate of the target-sample correction}
\label{app:unmix-sensitivity-estimation-proof}

The following lemma verifies the rate for
\eqref{eq:unmix-sensitivity-estimator}. The result holds when the
preliminary effect estimate and derivative averages use the same
target units.

\begin{lemma}[Rate of the target-sample correction]
\label{lem:unmix-sensitivity-estimation}
Assume the sampling, randomization, correctly specified HMM, mixture,
rank, history-score, and outcome moment conditions of
Theorem~\ref{thm:orthogonal-unmixing}. Let the HMM fit on its
independent history sample satisfy
\(\norm{\widehat\eta-\eta_0}=O_p(r_\eta)\) for
\(r_\eta=o(1)\). In neighborhoods of
\((\taub_0,\eta_0)\), suppose nonnegative envelopes
\(L_q(H)\), \(L_G(W)\), and \(L_J(H)\) satisfy
\(\E[L_q^2]+\E[L_G]+\E[L_J]<\infty\) and
\begin{align*}
  \norm{\qb_\eta(H)-\qb_0(H)}
  &\leq L_q(H)\norm{\eta-\eta_0},\\
  \opnorm{D_\eta\bm{\psi}(\taub,\eta)
           -D_\eta\bm{\psi}(\taub_0,\eta_0)}
  &\leq L_G(W)
      \{\norm{\taub-\taub_0}
        +\norm{\eta-\eta_0}\},\\
  \opnorm{D_\eta\bm{s}^H(H;\eta)
           -D_\eta\bm{s}^H(H;\eta_0)}
  &\leq L_J(H)\norm{\eta-\eta_0}.
\end{align*}
Assume also
\(\E[\fronorm{D_\eta\bm{\psi}_i
      (\taub_0,\eta_0)}^2]<\infty\) and
\(\E[\fronorm{D_\eta\bm{s}^H
      (H_i;\eta_0)}^2]<\infty\).
Then the construction in
\eqref{eq:unmix-sensitivity-estimator}, computed on all target
units, satisfies
\begin{equation}
  \norm{\widehat{\taub}^{\,\mathrm{pl}}-\taub_0}
    =O_p(n^{-1/2}+r_\eta),
  \qquad
  \opnorm{\widehat B-B_0}
    =O_p(n^{-1/2}+r_\eta).
  \label{eq:app-sensitivity-rate}
\end{equation}
\end{lemma}

\begin{proof}
Use the externally trained HMM fit after its training-sample label
rule; its labels agree with the true labeling with probability tending
to one. Write \(\widehat G=\E_n[D_\eta\bm{\psi}
  (\widehat{\taub}^{\,\mathrm{pl}},\widehat\eta)]\) and
\(\widehat J=\E_n[D_\eta\bm{s}^H(H;\widehat\eta)]\), so that
\(\widehat B=\widehat G\widehat J^{-1}\).
The envelope conditions justify
\(G_0=\E[D_\eta\bm{\psi}_i(\taub_0,\eta_0)]\).
The mixture identity and
\eqref{eq:unmix-randomization} give
\begin{equation*}
  \E[\qb_0(H_i)\widetilde A_iY_i]
  =v_A\E[\qb_0(H_i)\delta(H_i)]
  =M_0\taub_0.
\end{equation*}
Write \(g_\eta(W)=\widetilde A^2\qb_\eta(H)
\qb_\eta(H)^\top\). The posterior-simplex bound gives
\(\opnorm{g_\eta-g_{\eta_0}}
 \leq2L_q(H)\norm{\eta-\eta_0}\).
Decompose
\begin{equation*}
  \widehat M-M_0
  =(\E_n-\E)[g_{\eta_0}]
    +\E_n[g_{\widehat\eta}-g_{\eta_0}].
\end{equation*}
Because \(g_{\eta_0}\) is bounded and \(\E_n[L_q]=O_p(1)\), the
two terms are respectively \(O_p(n^{-1/2})\) and \(O_p(r_\eta)\).
Therefore
\begin{equation}
  \opnorm{\widehat M-M_0}
   =O_p(n^{-1/2}+r_\eta).
  \label{eq:app-sensitivity-preliminary-gram}
\end{equation}
Similarly, \(\qb_0\widetilde A Y\) has finite second moment,
while Cauchy--Schwarz gives
\(\E_n[L_q(H)|Y|]=O_p(1)\). The decomposition
\begin{align*}
  \E_n[\qb_{\widehat\eta}(H)\widetilde A Y]
       -M_0\taub_0
  &=(\E_n-\E)[\qb_0(H)\widetilde A Y]\\
  &\quad+\E_n[(\qb_{\widehat\eta}(H)-\qb_0(H))
                            \widetilde A Y]
\end{align*}
then gives
\begin{equation}
  \norm{\E_n[\qb_{\widehat\eta}(H)
                      \widetilde A Y]
        -M_0\taub_0}
   =O_p(n^{-1/2}+r_\eta).
  \label{eq:app-sensitivity-preliminary-rhs}
\end{equation}
Positive definiteness of \(M_0\) and
\eqref{eq:app-sensitivity-preliminary-gram}--
\eqref{eq:app-sensitivity-preliminary-rhs} imply the preliminary
effect rate in \eqref{eq:app-sensitivity-rate}.

For the derivative mean, add and subtract its value at the fixed
population arguments:
\begin{align*}
  \widehat G-G_0
  &=(\E_n-\E)[D_\eta\bm{\psi}
       (\taub_0,\eta_0)]+\E_n[D_\eta\bm{\psi}
       (\widehat{\taub}^{\,\mathrm{pl}},\widehat\eta)
       -D_\eta\bm{\psi}
       (\taub_0,\eta_0)].
\end{align*}
The first term is \(O_p(n^{-1/2})\) by its finite second moment.
The second is bounded in operator norm by
\(\E_n[L_G\{\norm{\widehat{\taub}^{\,\mathrm{pl}}-\taub_0}
 +\norm{\widehat\eta-\eta_0}\}]\), which is
\(O_p(n^{-1/2}+r_\eta)\) by the envelope law of large numbers.
The same decomposition with \(D_\eta\bm{s}^H\) and \(L_J\) gives
\(\opnorm{\widehat J-J_0}
 =O_p(n^{-1/2}+r_\eta)\).
These pointwise envelope bounds hold even though
\(\widehat{\taub}^{\,\mathrm{pl}}\) and the derivative averages use the same target
units. Nonsingularity of \(J_0\)
implies
\(\widehat J^{-1}=J_0^{-1}
  +O_p(n^{-1/2}+r_\eta)\).
Expanding \(\widehat G\widehat J^{-1}
 -G_0J_0^{-1}\) proves the asserted rate for
\(\widehat B\).
\end{proof}

\subsection{Orthogonal Correction and Inference}

\begin{proof}[Proof of Theorem~\ref{thm:orthogonal-unmixing}]
Under the additional conditions of
Lemma~\ref{lem:unmix-sensitivity-estimation}, the target-sample
construction satisfies the theorem's correction-rate assumption with
\(r_B=n^{-1/2}+r_\eta\). The proof below uses the theorem's stated
rate \(r_B\), which also allows other constructions that attain it.

In expectations, abbreviate
\(\qb_\eta(H_i)\) and \(\qb_0(H_i)\) by \(\qb_\eta\) and \(\qb_0\).
Correct HMM specification gives
\(\qb_0(H)=(\Pr(Z_i=s\mid H_i=H))_{s=1}^K\).
Assumption~\ref{assum:effect_sufficiency} implies
\(\delta(H)=\qb_0(H)^\top\taub_0\). Hence
\(\E[\qb_0(H_i)\delta(H_i)]=Q_0\taub_0\).
Positive definiteness of \(Q_0\) makes \(\taub_0\) the unique
coefficient satisfying this equation.

\paragraph{Population orthogonality.}
By \eqref{eq:unmix-population-moment},
\(\E[\bm{\psi}_i(\taub_0,\eta_0)]=\mathbf 0\).
The HMM score has mean zero at \(\eta_0\) by assumption. The negative
target Jacobian of the corrected score is
\begin{equation*}
  -\E[D_{\taub}\bm{\phi}_i(\taub_0,\eta_0,B_0)]
  =\E[\widetilde A_i^2\qb_0\qb_0^\top]
  =v_AQ_0=M_0.
\end{equation*}

Put \(u=\eta-\eta_0\). The local mean expansions
\eqref{eq:unmix-local-smoothness} give
\(\E[\bm{\psi}_i(\taub_0,\eta)]=G_0u+R_\psi(u)\) and
\(\E[\bm{s}^H(H_i;\eta)]=J_0u+R_s(u)\), where
\(\norm{R_\psi(u)}+\norm{R_s(u)}\leq C\norm{u}^2\).
For \(D=B-B_0\), the mean corrected score at \(\taub_0\) is
\begin{align*}
  \Phi(\eta,B)
  &:=\E[\bm{\phi}_i(\taub_0,\eta,B)]\\
  &=(G_0-B_0J_0)u
      +R_\psi(u)-B_0R_s(u)-D\{J_0u+R_s(u)\}.
\end{align*}
The first term vanishes by \(B_0=G_0J_0^{-1}\).
On a sufficiently small neighborhood,
\begin{equation*}
  \norm{\Phi(\eta,B)}
  \leq C\{\norm{\eta-\eta_0}^2
       +\opnorm{B-B_0}\norm{\eta-\eta_0}\}.
\end{equation*}

\paragraph{Target-sample fluctuation.}
The external HMM histories determine the labeling of
\(\widehat\eta\) independently of the target units. The assumed
rate for \(\widehat\eta-\eta_0\) is in this common labeling.
Put
\begin{align*}
  f_0(W)
    &=\bm{\phi}(W;\taub_0,\eta_0,B_0),\\
  d(W)
    &=\bm{\psi}(W;\taub_0,\widehat\eta)
      -B_0\bm{s}^H(H;\widehat\eta)-f_0(W).
\end{align*}
Since \(f_0(W)=\bm{\psi}(W;\taub_0,\eta_0)
-B_0\bm{s}^H(H;\eta_0)\), the triangle inequality gives
\begin{align*}
  \norm{d}_{L_2}
  &\leq\norm{\bm{\psi}(\taub_0,\widehat\eta)
              -\bm{\psi}(\taub_0,\eta_0)}_{L_2}\\
  &\quad+\opnorm{B_0}
       \norm{\bm{s}^H(H;\widehat\eta)
              -\bm{s}^H(H;\eta_0)}_{L_2}
   =o_p(1),
\end{align*}
by the assumed \(L_2\) continuity and consistency of
\(\widehat\eta\). Both terms in \(f_0\) have zero population mean, so
\(\E[f_0]=\mathbf 0\). Taking expectations in the definition of \(d\)
and setting \(u=\widehat\eta-\eta_0\), the local mean expansions give
\begin{align*}
  \E[d]
  &=\E[\bm{\psi}(\taub_0,\widehat\eta)]
      -B_0\E[\bm{s}^H(H;\widehat\eta)]\\
  &=\{G_0u+R_\psi(u)\}-B_0\{J_0u+R_s(u)\}\\
  &=R_\psi(u)-B_0R_s(u),
\end{align*}
where the last equality uses \(B_0J_0=G_0\). The remainder bound and
\(\norm{u}=O_p(r_\eta)\) therefore give
\(\norm{\E[d]}=O_p(r_\eta^2)\).
Let \(\mathcal F_{\mathrm{tr}}\) contain the external HMM training
data and fitting randomness. Conditional on \(\mathcal F_{\mathrm{tr}}\),
\(d\) is fixed and the target units are independent. Put
\(X_i=d(W_i)-\E[d]\), so
\(\E[X_i\mid\mathcal F_{\mathrm{tr}}]=\mathbf 0\).
Conditional independence makes the cross terms in the squared norm
vanish. Expanding \(\norm{X_i}^2\) then gives
\begin{align*}
  \E\!\left[
    \norm{(\E_n-\E)[d]}^2\mid\mathcal F_{\mathrm{tr}}
    \right]
  &=\frac1{n^2}\sum_{i=1}^n
    \E[\norm{X_i}^2\mid\mathcal F_{\mathrm{tr}}]\\
  &=\frac1n\{\norm{d}_{L_2}^2-\norm{\E[d]}^2\}\\
  &\leq\frac{\norm{d}_{L_2}^2}{n}.
\end{align*}
Conditional Chebyshev's inequality therefore gives
\((\E_n-\E)[d]=o_p(n^{-1/2})\).

The target estimating equation has sample mean
\(S_n(\taub)=\E_n[\bm{\psi}(W;\taub,\widehat\eta)]
  -\widehat B \E_n[\bm{s}^H(H;\widehat\eta)]\).
The estimate \(\widehat B\) uses the target units, so we separate
its error rather than condition on it. Adding and subtracting
\(B_0\E_n[\bm{s}^H(H;\widehat\eta)]\) gives
\begin{align*}
  S_n(\taub_0)
  &=\frac1n\sum_{i=1}^n\bm{\phi}_{0i}
    +\E[d]+(\E_n-\E)[d]\\
  &\quad-(\widehat B-B_0)
              \E_n[\bm{s}^H(H;\widehat\eta)].
\end{align*}
The second expansion in \eqref{eq:unmix-local-smoothness} gives
\(\E[\bm{s}^H(H;\widehat\eta)]=O_p(r_\eta)\).
Conditional on the external HMM fit, the target units are
independent. \(L_2\) continuity of the history score implies
\(\norm{\bm{s}^H(H;\widehat\eta)}_{L_2}=O_p(1)\), so conditional
Chebyshev's inequality gives
\((\E_n-\E)[\bm{s}^H(H;\widehat\eta)]=O_p(n^{-1/2})\).
Therefore
\((\widehat B-B_0)\E_n[\bm{s}^H(H;\widehat\eta)]
 =O_p\{r_B(n^{-1/2}+r_\eta)\}\).
Since \(r_B=o(1)\),
\begin{equation}
  S_n(\taub_0)
  =\frac1n\sum_{i=1}^n\bm{\phi}_{0i}
    +O_p(r_\eta^2+r_\eta r_B)+o_p(n^{-1/2}).
  \label{eq:app-unmix-target-score-expansion}
\end{equation}

\paragraph{Gram matrix, rate, and limit.}
We establish Gram consistency under the theorem's \(L_2\)
continuity condition. The preceding lemma gives a stronger
quantitative bound under its additional envelope condition.
Posterior probabilities are bounded, which keeps the sample Gram
matrix stable as \(\widehat\eta\) approaches \(\eta_0\).
Put \(g_\eta(W)=\widetilde A^2\qb_\eta(H)\qb_\eta(H)^\top\).
Posterior vectors lie in the probability simplex, so
\(\norm{\qb_\eta(H)}\leq1\) and
\begin{equation}
  \opnorm{g_\eta(W)-g_{\eta_0}(W)}
  \leq2\norm{\qb_\eta(H)-\qb_0(H)}.
  \label{eq:app-unmix-gram-bound}
\end{equation}
Decompose
\begin{align*}
  \widehat M-M_0
  &=(\E_n-\E)[g_{\eta_0}]
    +\E_n[g_{\widehat\eta}-g_{\eta_0}].
\end{align*}
The first term is \(O_p(n^{-1/2})\), since
\(g_{\eta_0}\) is bounded. Correct specification and the common state
labeling give \(\qb_{\eta_0}=\qb_0\). The assumed HMM rate has
\(r_\eta\to0\), so \(\widehat\eta\xrightarrow{p}\eta_0\). For any
\(\veps>0\), the \(L_2\) continuity of \(\qb_\eta\) supplies
\(\gamma>0\) such that
\(\norm{\eta-\eta_0}<\gamma\) implies
\(\norm{\qb_\eta-\qb_0}_{L_2}<\veps\). Thus
\begin{equation*}
  \Pr\!\left\{\norm{\qb_{\widehat\eta}-\qb_0}_{L_2}
          \geq\veps\right\}
  \leq\Pr\!\left\{\norm{\widehat\eta-\eta_0}\geq\gamma\right\}
  \to0,
\end{equation*}
which proves
\(\norm{\qb_{\widehat\eta}-\qb_0}_{L_2}=o_p(1)\).
By \eqref{eq:app-unmix-gram-bound}, the second term is bounded in operator norm by
\(2\E_n[\norm{\qb_{\widehat\eta}(H)-\qb_0(H)}]\). Conditional on the external fit, its expectation is at most
\(2\norm{\qb_{\widehat\eta}-\qb_0}_{L_2}=o_p(1)\).
Conditional Markov's inequality therefore gives
\(\widehat M=M_0+o_p(1)\). Since
\(\lambda_{\min}(M_0)\geq v_A\kappa_0\), Weyl's inequality implies,
with probability tending to one,
\begin{equation}
  \opnorm{\widehat M^{-1}}
    \leq\frac{2}{v_A\kappa_0},
  \qquad \widehat M^{-1}=M_0^{-1}+o_p(1).
  \label{eq:app-unmix-inverse}
\end{equation}
The matrix \(\widehat B\) is computed using the preliminary
\(\widehat{\taub}^{\,\mathrm{pl}}\) and held fixed when solving for
\(\taub\). By the definitions of \(\bm{\psi}\) and \(\widehat M\),
the sample score is affine in \(\taub\):
\begin{align*}
  S_n(\taub)
  &=\E_n[\qb_{\widehat\eta}(H)\widetilde A Y]
    -\widehat M\taub
    -\widehat B \E_n[\bm{s}^H(H;\widehat\eta)]\\
  &=S_n(\taub_0)-\widehat M(\taub-\taub_0).
\end{align*}
Equation~\eqref{eq:unmix-plugin} makes the first term
\(\widehat M\widehat{\taub}^{\,\mathrm{pl}}\), so
\eqref{eq:unmix-orthogonal-estimator} satisfies
\(S_n(\widehat{\taub}^{\,\perp})=\mathbf 0\). Thus
\(\mathbf 0=S_n(\taub_0)
-\widehat M(\widehat{\taub}^{\,\perp}-\taub_0)\), and inversion gives
the exact identity
\begin{equation}
  \widehat{\taub}^{\,\perp}-\taub_0
  =\widehat M^{-1}S_n(\taub_0).
  \label{eq:app-unmix-linear-equation}
\end{equation}
The mean-zero vectors \(\bm{\phi}_{0i}\) have finite variance because
posterior probabilities and \(\widetilde A_i\) are bounded, while
\(Y_i(a)\) and the true history score have finite
second moments. Their average is therefore \(O_p(n^{-1/2})\). Combining
\eqref{eq:app-unmix-target-score-expansion} and
\eqref{eq:app-unmix-inverse}--\eqref{eq:app-unmix-linear-equation}
gives the expansion \eqref{eq:unmix-orthogonal-expansion}. The inverse
bound in \eqref{eq:app-unmix-inverse} then gives
\eqref{eq:unmix-rate}. If
\(\sqrt n(r_\eta^2+r_\eta r_B)\to0\), the remainder in
\eqref{eq:unmix-orthogonal-expansion} is \(o_p(n^{-1/2})\).
The normal limit in \eqref{eq:unmix-clt} follows from the multivariate
central limit theorem and Slutsky's theorem.

\paragraph{Sandwich covariance.}
Let \(\widehat{\bm{\phi}}_i\) be the fitted score at
\(\widehat{\taub}^{\,\perp}\). Affinity gives
\begin{equation*}
  \widehat{\bm{\phi}}_i-\bm{\phi}_{0i}
  =d(W_i)-(\widehat B-B_0)
       \bm{s}^H(H_i;\widehat\eta)
       -\widetilde A_i^2\widehat{\qb}_i
       \widehat{\qb}_i^\top
       (\widehat{\taub}^{\,\perp}-\taub_0).
\end{equation*}
Conditional Markov's inequality and
\(\norm{d}_{L_2}=o_p(1)\) imply
\(\E_n[\norm{d(W)}^2]=o_p(1)\).
The \(L_2\) continuity of the history score and independence of
the external fit imply
\(\E_n[\norm{\bm{s}^H(H;\widehat\eta)}^2]=O_p(1)\).
This, \(\opnorm{\widehat B-B_0}=o_p(1)\), the posterior-simplex
bound, and the estimation rate imply
\(n^{-1}\sum_i\norm{\widehat{\bm{\phi}}_i
                 -\bm{\phi}_{0i}}^2=o_p(1)\).
Cauchy--Schwarz and the law of large numbers for
\(n^{-1}\sum_i\bm{\phi}_{0i}\bm{\phi}_{0i}^\top\) give
\(\widehat\Omega\xrightarrow{p}\Omega_0\).
Together with \eqref{eq:app-unmix-inverse}, this proves the sandwich
assertion in the theorem.
\end{proof}

\section{When treatment-effect sufficiency is violated}
\label{app:effect-sufficiency-violation}

This section keeps randomized treatment and a correctly specified,
label-aligned HMM with the true number of states \(K\), but allows
treatment effects to vary with history within a state. If the estimator in
Section~\ref{sec:unmixing-estimation} is left unchanged, it estimates a
projection of the history-conditional effect. The gap from the true state
effects depends on how within-state effect variation covaries with the
posterior; the prevalence-weighted average still equals the overall ATE.

\paragraph{Conditional effects and the projection target.}
Let \(\Delta_i=Y_i(1)-Y_i(0)\). For each state with positive
prevalence, write
\[
  \tau_s(h)=\E[\Delta_i\mid Z_i=s,H_i=h],\qquad
  \tau_{0s}=\E[\tau_s(H_i)\mid Z_i=s].
\]
The conditional effect \(\tau_s(h)\) is needed only where
\(q_{0s}(h)>0\).
Define the residual and decompose the conditional effect:
\begin{equation}
\begin{aligned}
  r(H)
  &=\sum_{s=1}^K q_{0s}(H)\{\tau_s(H)-\tau_{0s}\},\\
  \delta(H)
  &=\sum_{s=1}^K q_{0s}(H)\tau_s(H)
    =\qb_0(H)^\top\taub_0+r(H).
\end{aligned}
\label{eq:app-sufficiency-decomposition}
\end{equation}
Assumption~\ref{assum:effect_sufficiency} makes every within-state
deviation zero and therefore implies \(r(H)=0\). The converse need not
hold: deviations in different states can cancel in the observable
mixture.

With finite outcome second moments, suppose
\(Q_0=\E[\qb_0(H_i)\qb_0(H_i)^\top]\) is nonsingular. The
population moment \eqref{eq:unmix-population-moment} then has the unique
root
    \begin{equation}
      \taub^\star
      =Q_0^{-1}\E[\qb_0(H_i)\delta(H_i)]
      =\argmin_{\taub\in\R^K}
        \E\!\left[\{\delta(H_i)-\qb_0(H_i)^\top\taub\}^2\right].
      \label{eq:app-sufficiency-projection}
    \end{equation}
Thus \(\qb_0(H)^\top\taub^\star\) is the best linear approximation to
\(\delta(H)\) using the state posterior coordinates. 

\paragraph{Bias relative to the state effects.}
For any measurable \(g(H_i)\) integrable within state \(s\),
posterior reweighting gives
\begin{equation}
  \E[q_{0s}(H_i)g(H_i)]
  =\pi_s\,\E[g(H_i)\mid Z_i=s].
  \label{eq:app-sufficiency-reweighting}
\end{equation}
This identity follows from
\(q_{0s}(H_i)=\E[\ind\{Z_i=s\}\mid H_i]\).

\begin{proposition}[Bias under within-state effect variation]
\label{prop:unmix-effect-variation}
Suppose \(A_i\sim\operatorname{Bernoulli}(e)\) independently of
\((H_i,Y_i(0),Y_i(1))\), with known \(e\in(0,1)\). Suppose outcomes are
consistent with finite second moments, the HMM is correctly specified
with the true \(K\), and \(\lambda_{\min}(Q_0)>0\). Treatment-effect
sufficiency is not required. Then
\begin{equation}
  \begin{aligned}
    c_t&=\sum_{s=1}^K\pi_s
      \operatorname{Cov}\{q_{0t}(H_i),\tau_s(H_i)\mid Z_i=s\},
      \quad t=1,\ldots,K,\\
    \taub^\star-\taub_0&=Q_0^{-1}\bm{c}, \quad \bm{c}=(c_1,\ldots,c_K)^\top.
  \end{aligned}
  \label{eq:app-sufficiency-bias}
\end{equation}
Moreover,
\[
  \bm{\pi}^\top\taub^\star
  =\bm{\pi}^\top\taub_0
  =\E[\Delta_i].
\]
\end{proposition}

\begin{proof}
Equations~\eqref{eq:app-sufficiency-decomposition} and
\eqref{eq:app-sufficiency-projection} give
\[
  Q_0(\taub^\star-\taub_0)
  =\E[\qb_0(H_i)r(H_i)].
\]
For coordinate \(t\), apply
\eqref{eq:app-sufficiency-reweighting} to each term in \(r(H_i)\):
\[
  \E[q_{0t}(H_i)r(H_i)]
  =\sum_{s=1}^K\pi_s
    \E\!\left[
      q_{0t}(H_i)\{\tau_s(H_i)-\tau_{0s}\}
      \mid Z_i=s
    \right].
\]
Because \(\E[\tau_s(H_i)\mid Z_i=s]=\tau_{0s}\), each conditional
expectation is the covariance in \eqref{eq:app-sufficiency-bias}.
Finite outcome second moments make these terms integrable.

The posterior coordinates sum to one, so
\(Q_0\mathbf 1=\E[\qb_0(H_i)]=\bm{\pi}\). Hence
\[
  \bm{\pi}^\top\taub^\star
  =\mathbf 1^\top Q_0\taub^\star
  =\E[\delta(H_i)]
  =\E[\Delta_i]
  =\bm{\pi}^\top\taub_0.
\]
\end{proof}

Within-state effect variation alone does not force bias. A sufficient
condition for \(\taub^\star=\taub_0\) is that, within every true state,
\(\tau_s(H_i)\) is uncorrelated with each posterior coordinate. Bias can
arise when histories that make a state more likely also change treatment
response among members of that state. For example, members with a
large \(q_{0s}(H)\) might respond more strongly than less typical
members of the same true state.

Because \(\opnorm{Q_0^{-1}}=1/\lambda_{\min}(Q_0)\), weak posterior rank
can amplify the covariance vector in \eqref{eq:app-sufficiency-bias}.
The resulting error redistributes effect across states: some state
effects can be too high and others too low, but their prevalence-weighted
average remains the overall ATE.

\paragraph{A sensitivity bound from posterior spread.}
Let
\(\sigma_{ts}=\operatorname{std}\big(q_{0t}(H_i)\mid Z_i=s\big)\).
Equation~\eqref{eq:app-sufficiency-reweighting} identifies this
within-state posterior spread from the distribution of the true
posteriors:
\[
  \sigma_{ts}^2
  =\frac{\E[q_{0s}(H_i)q_{0t}(H_i)^2]}{\pi_s}
   -\left\{
      \frac{\E[q_{0s}(H_i)q_{0t}(H_i)]}{\pi_s}
    \right\}^{\!2}.
\]
Suppose the analyst posits
\(\operatorname{std}\big(\tau_s(H_i)\mid Z_i=s\big)\leq\Gamma\) for every state.
Cauchy--Schwarz applied to \eqref{eq:app-sufficiency-bias} gives
\begin{equation}
  \begin{aligned}
    b_k&=\sum_{t=1}^K |(Q_0^{-1})_{kt}|
      \sum_{s=1}^K\pi_s\sigma_{ts},\\
    |\tau_k^\star-\tau_{0k}|&\leq\Gamma b_k.
  \end{aligned}
  \label{eq:app-sufficiency-bound}
\end{equation}
All terms in \(b_k\) depend only on the posterior distribution and can
be estimated from fitted target posteriors. The bound is conservative:
it allows the covariance terms to attain their largest magnitudes with
the signs that maximize the error in state \(k\).

\section{When the fitted number of states is incorrect}
\label{app:wrong-K}

The main identification and inference results assume a correctly specified
HMM with the true number of states. Here we describe the population target
of posterior unmixing when the fitted state count differs from the true
count. The same characterization applies to other misspecified state models.

Let \(K_0\) be the true state count and \(K\) the fitted count. The
main-text posterior \(\qb_0\), effect vector \(\taub_0\), and
mean posterior \(\bm{\pi}\) have
\(K_0\) coordinates here; the mixture identity \eqref{eqn:mixture}
continues to hold. We retain randomized treatment, consistency, and
\(\E[Y_i(a)^2]<\infty\) for \(a\in\{0,1\}\).

Suppose an HMM fitted on independent histories yields, after a fixed
label rule, posterior maps converging in \(L_2\) to a deterministic
map \(\qb^*(H)=(q_1^*(H),\ldots,q_K^*(H))^\top\) whose coordinates
are nonnegative and sum to one:
\[
  \norm{\qb_{\widehat\eta}-\qb^*}_{L_2}\to_p0.
\]
Define the fitted posterior Gram matrix and its mean:
\[
  Q^*=\E[\qb^*(H_i)\qb^*(H_i)^\top],\qquad
  \bm{\pi}^*=\E[\qb^*(H_i)].
\]
Assume \(Q^*\) is nonsingular.
We impose no merge or refinement relation between \(\qb^*\) and
\(\qb_0\).

\begin{proposition}[Projection target and overall average effect]
\label{prop:wk-blp}
The plug-in estimator in \eqref{eq:unmix-plugin} converges to
\begin{equation}
  \taub^\star=(Q^*)^{-1}\E[\qb^*(H_i)\delta(H_i)]
  =\argmin_{\taub\in\R^K}
     \E\bigl[\{\delta(H_i)-\qb^*(H_i)^\top\taub\}^2\bigr].
  \label{eq:wk-target}
\end{equation}
Thus \(\qb^*(H)^\top\taub^\star\) is the least-squares approximation
of \(\delta(H)\) from the fitted posterior coordinates, and
\begin{equation}
  (\bm{\pi}^*)^\top\taub^\star
  =\E[Y_i(1)-Y_i(0)].
  \label{eq:wk-overall-ate}
\end{equation}
Under Assumption~\ref{assum:effect_sufficiency}, define
\[
  \Lambda=(Q^*)^{-1}
       \E[\qb^*(H_i)\qb_0(H_i)^\top]
       \in\R^{K\times K_0}.
\]
Then \(\taub^\star=\Lambda\taub_0\) and
\(\Lambda\mathbf{1}_{K_0}=\mathbf{1}_K\).
The rows of \(\Lambda\) need not be nonnegative.
\end{proposition}

\begin{proof}
By \eqref{eq:unmix-plugin}, the assumed \(L_2\) convergence, and the
randomization identities in \eqref{eq:unmix-randomization}, conditional
laws of large numbers give
\[
  \widehat M\to_p v_AQ^*,\qquad
  \frac1n\sum_{i=1}^n\widehat{\qb}_i\widetilde A_iY_i
    \to_p v_A\E[\qb^*(H_i)\delta(H_i)].
\]
Since \(Q^*\) is nonsingular, the estimator converges to
\(\taub^\star\) in \eqref{eq:wk-target}. The least-squares claim follows
from the normal equation
\(Q^*\taub=\E[\qb^*(H_i)\delta(H_i)]\).

Posterior coordinates sum to one, so
\(Q^*\mathbf{1}_K=\bm{\pi}^*\). Thus
\[
  (\bm{\pi}^*)^\top\taub^\star
  =\mathbf{1}_K^\top\E[\qb^*(H_i)\delta(H_i)]
  =\E[\delta(H_i)]=\E[Y_i(1)-Y_i(0)].
\]
Substituting \(\delta(H_i)=\qb_0(H_i)^\top\taub_0\) from
\eqref{eqn:mixture} into \eqref{eq:wk-target} gives
\[
  \taub^\star
  =(Q^*)^{-1}\E[\qb^*(H_i)\qb_0(H_i)^\top]\taub_0
  =\Lambda\taub_0.
\]
Since
\(\qb_0^\top\mathbf{1}_{K_0}=1\),
\(\Lambda\mathbf{1}_{K_0}=(Q^*)^{-1}\bm{\pi}^*
=\mathbf{1}_K\).
\end{proof}

The weighting in \eqref{eq:wk-overall-ate} uses the \emph{mean fitted
posterior} \(\bm{\pi}^*\). The true posterior mean \(\bm{\pi}\) is the
true-state prevalence vector. Under misspecification, the fitted mean need not
represent true-state prevalences and may even have a different dimension.
Likewise, \(\Lambda\) is a linear coefficient map, not necessarily a matrix of probabilities.  Below is an example that gives negative entries of \(\Lambda\).

\begin{example}[Negative coefficient weights]
\label{ex:wk-negative}
Let \(h(H)\sim\operatorname{Unif}(0,1)\), and consider the valid
simplex-valued maps
\[
  \qb^*=(h,1-h)^\top,\qquad
  \qb_0=(h^2,\,2h(1-h),\,(1-h)^2)^\top.
\]
Linear combinations of \(q_1^*=h\) and \(q_2^*=1-h\) are affine in
\(h\). Regressing the coordinates of \(\qb_0\) on these fitted
posterior coordinates gives
\[
  \Lambda=
  \begin{pmatrix}
    5/6 & 1/3 & -1/6\\
   -1/6 & 1/3 & 5/6
  \end{pmatrix}.
\]
For \(\taub_0=(0,0,1)^\top\), the fitted coefficients are
\(\taub^\star=(-1/6,5/6)^\top\): one is negative although all true-state
effects are nonnegative. The least-squares reconstruction of \(\qb_0\)
is not itself a posterior: its first coordinate is
\(h-1/6<0\) for \(h<1/6\).
\end{example}

\section{Oracle-posterior simulation details}
\label{app:simulations}

\subsection{Data-generating process and metrics}
\label{app:simulation-design}

Let \(K=3\), draw \(Z_i\in\{1,2,3\}\) with state probabilities
\(\bm{\pi}=(\pi_1,\pi_2,\pi_3)^\top\), and set
\(\taub_0=(-1,0,1)^\top\). The pre-treatment history is a single
categorical signal \(H_i\in\{1,2,3\}\), with state-conditional accuracy
\(\rho=\PP(H_i=Z_i\mid Z_i=s)\) common to every state \(s\):
\[
\Pr(H_i=h\mid Z_i=s)=
\begin{cases}
\rho, & h=s,\\
(1-\rho)/2, & h\ne s.
\end{cases}
\]
This signal creates ambiguity about \(Z_i\) even though its posterior
is known exactly. By iterated expectation, its mean posterior is the
state prior: \(\E[\qb_0(H_i)]=\bm{\pi}\). No HMM trajectory or HMM
parameter is simulated.

Independently across units, draw
\(\veps_i\sim\mathcal N(0,1)\), independently of \((Z_i,H_i)\), and
randomize \(A_i\sim\operatorname{Bernoulli}(e)\) independently of these
variables, with \(e=1/2\). With
\(\widetilde A_i=A_i-e\), the outcome is
\[
Y_i=\widetilde A_i\tau_{0,Z_i}+\veps_i.
\]
Thus \(Y_i(a)=(a-e)\tau_{0,Z_i}+\veps_i\),
\(\E[Y_i(1)-Y_i(0)\mid Z_i=s,H_i]=\tau_{0s}\), and
\(\E[Y_i\mid H_i]=0\). Hard and soft contrasts use the exact
\(\qb_0(H_i)\); unmixing uses \eqref{eq:unmix-plugin} directly.
We also report an
unattainable true-state oracle: it observes each simulated \(Z_i\)
and takes the treated-minus-control mean within each state.

For each condition, we use
\(n\in\{500,2{,}000,10{,}000,50{,}000\}\) and \(R=1{,}000\)
independent replications. For method-specific estimate
\(\widehat\tau_{rk}\) of \(\tau_{0k}\) in replication \(r\), the
primary metric gives each state equal weight:
\[
\operatorname{RMSE}_{\mathrm{state}}
=\left\{\frac{1}{RK}
\sum_{r=1}^{R}\sum_{k=1}^{K}
(\widehat\tau_{rk}-\tau_{0k})^2\right\}^{1/2}.
\]
We also examine state-specific bias and RMSE, interval width, coverage
of nominal \(95\%\) confidence intervals, and the condition number of
the population design matrix \(M_0\). Let
\(\widehat{\operatorname{SE}}_k\) be a method's estimated account-level
standard error for state \(k\). Each interval uses the normal approximation
\(\widehat\tau_k\pm1.96\widehat{\operatorname{SE}}_k\).
The oracle, hard, and soft methods use the empirical treated- and
control-arm variances of their weighted means. Unmixing uses the
empirical sandwich variance of its residualized-treatment moment.
Coverage is the proportion of state-by-replication intervals containing
the true \(\tau_{0k}\). Parenthetical Monte Carlo standard errors (SEs) in the
tables quantify uncertainty in these simulation summaries.

\subsection{Balanced states and posterior overlap}
\label{app:sim-overlap}

Set \(\pi_1=\pi_2=\pi_3=1/3\) and vary
\(\rho\in\{0.40,0.60,0.80,0.95,1.00\}\). For \(\rho>1/3\), the
maximum a posteriori (MAP)
state is \(H_i\), with classification accuracy \(\rho\). Bayes' rule
gives
\[
q_{0k}(h)=\Pr(Z_i=k\mid H_i=h)
=\begin{cases}
\rho, & k=h,\\
(1-\rho)/2, & k\ne h.
\end{cases}
\]
Let \(d_\rho=(3\rho-1)/2\). Applying
Proposition~\ref{prop:naive-estimator-bias} coordinatewise yields
\[
\begin{aligned}
\widehat{\taub}^{\,\mathrm{hard}}&\xrightarrow{p}
  d_\rho\taub_0,\\
\widehat{\taub}^{\,\mathrm{soft}}&\xrightarrow{p}
  d_\rho^2\taub_0.
\end{aligned}
\]
For \(\rho<1\), both estimators attenuate the nonzero state effects,
with greater attenuation under soft weighting. These errors persist
as \(n\) grows.

The population Gram matrix
\(M_0=v_A\E[\qb_0(H_i)\qb_0(H_i)^\top]\) has eigenvalues
\(v_A/3\) and \(v_Ad_\rho^2/3\), the latter with multiplicity two.
Consequently, the spectral condition number is
\[
\operatorname{cond}(M_0)
=d_\rho^{-2}
=\frac{4}{(3\rho-1)^2}.
\]
The inverse amplifies sampling noise as \(\rho\) approaches the
uninformative boundary \(1/3\). At \(\rho=0.40\), this condition number
is \(100\).

At \(\rho=0.40\), unmixing RMSE falls from 1.386 to 0.679, 0.309,
and 0.139 as \(n\) increases from 500 to 2,000, 10,000, and
50,000; hard-assignment RMSE remains near 0.74. Thus the
variance cost of unmixing is most visible at small \(n\).

Table~\ref{tab:overlap-results} shows coverage of the true state effects
at \(n=50{,}000\). Across all simulated sample sizes and signal
accuracies, state-specific unmixing coverage ranges from 93.3\% to
96.8\%, and its equal-state average ranges from 94.2\% to 95.9\%.
For \(\rho<1\) at \(n=50{,}000\), hard and soft intervals cover the true
effects only 31--32\% of the time, consistent with their bias for the
latent-state effects.
Even with the fairly accurate signal \(\rho=0.95\), hard and soft
true-effect coverage falls from 92.6\% and 86.8\% at \(n=500\)
to 31.5\% and 31.4\% at \(n=50{,}000\), as their intervals narrow
around the mixture limits.

At \(\rho=1\), \(H_i=Z_i\) and
\(q_{0k}(H_i)=\ind\{Z_i=k\}\). Hard and soft contrasts then use the
same true-state groups. Residualized unmixing can differ in a finite
sample because it uses the known randomization probability \(e\),
whereas a difference in means uses the realized arm counts. They are
asymptotically equivalent in this design. At \(n=50{,}000\), the three
RMSEs in Table~\ref{tab:overlap-results} and the true-state oracle RMSE
in Table~\ref{tab:overlap-rmse-full} all round to \(0.015\).

Tables~\ref{tab:overlap-rmse-full} and
\ref{tab:overlap-coverage-full} give the complete grids across sample
sizes and signal accuracies, with Monte Carlo standard errors.

\begin{table*}[t]
\centering
\caption{Equal-state RMSE in the oracle-posterior experiment, based on
1,000 replications per condition. Entries are RMSE (Monte Carlo
standard error). The oracle observes the true state; all other methods
use the exact Bayes posterior.}
\label{tab:overlap-rmse-full}
\small
\setlength{\tabcolsep}{4pt}
\begin{tabular}{rlcccc}
\toprule
\(n\) & Signal accuracy \(\rho\) & Oracle & Unmixing & Hard & Soft \\
\midrule
500 & 0.40 & 0.1515 (0.0020) & 1.3861 (0.020) & 0.7601 (0.0031) & 0.8143 (0.00037) \\
500 & 0.60 & 0.1552 (0.0020) & 0.3507 (0.0052) & 0.5165 (0.0029) & 0.6949 (0.0012) \\
500 & 0.80 & 0.1533 (0.0020) & 0.2074 (0.0027) & 0.2935 (0.0027) & 0.4382 (0.0020) \\
500 & 0.95 & 0.1562 (0.0019) & 0.1659 (0.0020) & 0.1700 (0.0021) & 0.1924 (0.0023) \\
500 & 1.00 & 0.1520 (0.0019) & 0.1513 (0.0019) & 0.1520 (0.0019) & 0.1520 (0.0019) \\
\addlinespace[4pt]
2,000 & 0.40 & 0.0774 (0.00100) & 0.6789 (0.011) & 0.7394 (0.0015) & 0.8097 (0.00016) \\
2,000 & 0.60 & 0.0791 (0.0011) & 0.1742 (0.0025) & 0.4994 (0.0015) & 0.6893 (0.00060) \\
2,000 & 0.80 & 0.0784 (0.0010) & 0.1040 (0.0014) & 0.2561 (0.0014) & 0.4206 (0.0010) \\
2,000 & 0.95 & 0.0759 (0.0010) & 0.0809 (0.0010) & 0.0995 (0.0012) & 0.1404 (0.0012) \\
2,000 & 1.00 & 0.0771 (0.00099) & 0.0771 (0.00099) & 0.0771 (0.00099) & 0.0771 (0.00099) \\
\addlinespace[4pt]
10,000 & 0.40 & 0.0350 (0.00045) & 0.3085 (0.0047) & 0.7366 (0.00071) & 0.8087 (0.000072) \\
10,000 & 0.60 & 0.0347 (0.00044) & 0.0782 (0.0012) & 0.4908 (0.00069) & 0.6861 (0.00028) \\
10,000 & 0.80 & 0.0347 (0.00046) & 0.0479 (0.00065) & 0.2472 (0.00065) & 0.4172 (0.00046) \\
10,000 & 0.95 & 0.0348 (0.00047) & 0.0373 (0.00048) & 0.0707 (0.00059) & 0.1226 (0.00057) \\
10,000 & 1.00 & 0.0346 (0.00044) & 0.0346 (0.00044) & 0.0346 (0.00044) & 0.0346 (0.00044) \\
\addlinespace[4pt]
50,000 & 0.40 & 0.0158 (0.00020) & 0.1393 (0.0021) & 0.7349 (0.00032) & 0.8084 (0.000032) \\
50,000 & 0.60 & 0.0152 (0.00020) & 0.0347 (0.00051) & 0.4897 (0.00029) & 0.6858 (0.00012) \\
50,000 & 0.80 & 0.0155 (0.00019) & 0.0214 (0.00028) & 0.2453 (0.00030) & 0.4165 (0.00021) \\
50,000 & 0.95 & 0.0150 (0.00021) & 0.0160 (0.00022) & 0.0633 (0.00027) & 0.1189 (0.00025) \\
50,000 & 1.00 & 0.0154 (0.00019) & 0.0154 (0.00019) & 0.0154 (0.00019) & 0.0154 (0.00019) \\
\bottomrule
\end{tabular}
\end{table*}

\begin{table*}[t]
\centering
\caption{Coverage of the true latent-state effects by nominal 95\%
intervals in the oracle-posterior experiment, based on 1,000
replications per condition. Entries are coverage (Monte Carlo standard
error), averaged equally across the three states.}
\label{tab:overlap-coverage-full}
\small
\setlength{\tabcolsep}{5pt}
\begin{tabular}{rlcccc}
\toprule
\(n\) & Signal accuracy \(\rho\) & Oracle & Unmixing & Hard & Soft \\
\midrule
500 & 0.40 & 0.955 (0.0038) & 0.949 (0.0044) & 0.314 (0.0024) & 0.323 (0.0018) \\
500 & 0.60 & 0.942 (0.0043) & 0.947 (0.0045) & 0.356 (0.0044) & 0.318 (0.0022) \\
500 & 0.80 & 0.947 (0.0041) & 0.953 (0.0038) & 0.671 (0.0078) & 0.330 (0.0032) \\
500 & 0.95 & 0.944 (0.0041) & 0.948 (0.0041) & 0.926 (0.0048) & 0.868 (0.0059) \\
500 & 1.00 & 0.956 (0.0037) & 0.956 (0.0037) & 0.956 (0.0037) & 0.956 (0.0037) \\
\addlinespace[4pt]
2,000 & 0.40 & 0.950 (0.0041) & 0.949 (0.0046) & 0.319 (0.0022) & 0.321 (0.0020) \\
2,000 & 0.60 & 0.939 (0.0044) & 0.950 (0.0042) & 0.314 (0.0025) & 0.312 (0.0026) \\
2,000 & 0.80 & 0.947 (0.0040) & 0.952 (0.0040) & 0.343 (0.0036) & 0.318 (0.0022) \\
2,000 & 0.95 & 0.954 (0.0038) & 0.954 (0.0038) & 0.884 (0.0057) & 0.649 (0.0076) \\
2,000 & 1.00 & 0.952 (0.0039) & 0.953 (0.0039) & 0.952 (0.0039) & 0.952 (0.0039) \\
\addlinespace[4pt]
10,000 & 0.40 & 0.945 (0.0042) & 0.949 (0.0045) & 0.320 (0.0021) & 0.316 (0.0024) \\
10,000 & 0.60 & 0.949 (0.0040) & 0.952 (0.0042) & 0.316 (0.0023) & 0.314 (0.0024) \\
10,000 & 0.80 & 0.945 (0.0041) & 0.943 (0.0043) & 0.311 (0.0026) & 0.312 (0.0026) \\
10,000 & 0.95 & 0.949 (0.0042) & 0.942 (0.0043) & 0.600 (0.0075) & 0.321 (0.0030) \\
10,000 & 1.00 & 0.948 (0.0040) & 0.949 (0.0040) & 0.948 (0.0040) & 0.948 (0.0040) \\
\addlinespace[4pt]
50,000 & 0.40 & 0.948 (0.0041) & 0.951 (0.0044) & 0.317 (0.0023) & 0.318 (0.0022) \\
50,000 & 0.60 & 0.956 (0.0038) & 0.953 (0.0042) & 0.317 (0.0023) & 0.316 (0.0023) \\
50,000 & 0.80 & 0.951 (0.0040) & 0.948 (0.0043) & 0.318 (0.0022) & 0.319 (0.0022) \\
50,000 & 0.95 & 0.957 (0.0038) & 0.959 (0.0036) & 0.315 (0.0025) & 0.314 (0.0025) \\
50,000 & 1.00 & 0.948 (0.0040) & 0.948 (0.0040) & 0.948 (0.0040) & 0.948 (0.0040) \\
\bottomrule
\end{tabular}
\end{table*}

\subsection{State imbalance with oracle posterior}
\label{app:sim-imbalance}

We fix \(\rho=0.80\) and vary the state prior \(\bm{\pi}\):
\[
\begin{aligned}
\bm{\pi}^{\,\mathrm{balanced}}&=(1/3,1/3,1/3)^\top,\\
\bm{\pi}^{\,\mathrm{moderate}}&=(0.50,0.30,0.20)^\top,\\
\bm{\pi}^{\,\mathrm{severe}}&=(0.70,0.20,0.10)^\top.
\end{aligned}
\]
The signal channel, treatment assignment, outcome model, and
\(\taub_0\) remain fixed. For each prior, we recompute the exact
posterior by Bayes' rule:
\[
q_{0k}(h)=
\frac{\pi_k\Pr(H_i=h\mid Z_i=k)}
{\sum_{\ell=1}^K\pi_\ell\Pr(H_i=h\mid Z_i=\ell)}.
\]
For all three priors, the MAP state remains \(H_i\), so classification
accuracy stays at \(0.80\). The designs therefore change state
prevalence and posterior geometry while holding signal accuracy fixed.
The marginal treatment effect also changes because
\(\bm{\pi}^\top\taub_0\) depends on the prior.

Alongside equal-state RMSE, we report prevalence-weighted RMSE:
\[
\operatorname{RMSE}_{\bm{\pi}}
=\left\{\frac1R\sum_{r=1}^R\sum_{k=1}^K
\pi_k(\widehat\tau_{rk}-\tau_{0k})^2\right\}^{1/2}.
\]
The first metric gives each state equal weight; the second gives more
weight to states that are more common in the target population.

Table~\ref{tab:imbalance-rmse} reports results at \(n=50{,}000\).
From the balanced to the severe prior,
\(\operatorname{cond}(M_0)\) increases from \(2.04\) to \(18.12\),
and unmixing's equal-state RMSE rises from \(0.021\) to \(0.037\).
Its coverage remains near \(95\%\). Hard and soft RMSE rise more
sharply as their mixture targets move away from the state effects.
Under severe imbalance, unmixing's equal-state RMSE decreases from
0.353 to 0.184, 0.079, and 0.037 across the four sample sizes,
while average true-effect coverage remains between 94.8\% and
95.5\%.

\begin{table}[t]
\centering
\caption{State imbalance at \(n=50{,}000\) and \(\rho=0.80\), based on
1,000 replications. The three method columns report equal-state RMSE;
coverage is for unmixing's nominal \(95\%\) intervals, averaged
across states.}
\label{tab:imbalance-rmse}
\small
\setlength{\tabcolsep}{5pt}
\begin{tabular}{lcccccc}
\toprule
State prior
& \(\min_k\pi_k\)
& \(\operatorname{cond}(M_0)\)
& Unmixing
& Hard
& Soft
& Coverage \\
\midrule
Balanced & 0.333 & 2.04  & 0.021 & 0.245 & 0.417 & 0.948 \\
Moderate & 0.200 & 4.49  & 0.024 & 0.331 & 0.474 & 0.947 \\
Severe   & 0.100 & 18.12 & 0.037 & 0.565 & 0.655 & 0.948 \\
\bottomrule
\end{tabular}
\end{table}

Table~\ref{tab:imbalance-weighted-rmse} reports prevalence-weighted
RMSE at every sample size, with Monte Carlo standard errors. Under
balance, the equal-state and prevalence-weighted metrics coincide by
definition. At \(n=50{,}000\), unmixing's prevalence-weighted RMSE
under severe imbalance is \(0.0244\), versus \(0.037\) when states are
weighted equally.

\begin{table*}[t]
\centering
\caption{Prevalence-weighted RMSE in the oracle-posterior state-imbalance
experiment at \(\rho=0.80\), based on 1,000 replications per condition.
Monte Carlo standard errors appear in parentheses. Each state's squared
error receives weight \(\pi_k\); the oracle observes the true state.}
\label{tab:imbalance-weighted-rmse}
\small
\setlength{\tabcolsep}{4pt}
\begin{tabular}{rlcccc}
\toprule
\(n\) & State prior & Oracle & Unmixing & Hard & Soft \\
\midrule
500 & Balanced & 0.1533 (0.0020) & 0.2074 (0.0027) & 0.2935 (0.0027) & 0.4382 (0.0020) \\
    & Moderate & 0.1548 (0.0021) & 0.2158 (0.0030) & 0.3146 (0.0026) & 0.4389 (0.0020) \\
    & Severe   & 0.1557 (0.0020) & 0.2358 (0.0035) & 0.3557 (0.0024) & 0.4312 (0.0018) \\
\addlinespace
2,000 & Balanced & 0.0784 (0.0010) & 0.1040 (0.0014) & 0.2561 (0.0014) & 0.4206 (0.0010) \\
      & Moderate & 0.0784 (0.0010) & 0.1083 (0.0014) & 0.2817 (0.0014) & 0.4225 (0.0010) \\
      & Severe   & 0.0785 (0.00098) & 0.1217 (0.0017) & 0.3312 (0.0012) & 0.4165 (0.00092) \\
\addlinespace
10,000 & Balanced & 0.0347 (0.00046) & 0.0479 (0.00065) & 0.2472 (0.00065) & 0.4172 (0.00046) \\
       & Moderate & 0.0347 (0.00045) & 0.0479 (0.00067) & 0.2730 (0.00061) & 0.4188 (0.00043) \\
       & Severe   & 0.0343 (0.00046) & 0.0528 (0.00073) & 0.3245 (0.00054) & 0.4129 (0.00040) \\
\addlinespace
50,000 & Balanced & 0.0155 (0.00019) & 0.0214 (0.00028) & 0.2453 (0.00030) & 0.4165 (0.00021) \\
       & Moderate & 0.0158 (0.00021) & 0.0220 (0.00030) & 0.2720 (0.00029) & 0.4185 (0.00021) \\
       & Severe   & 0.0159 (0.00021) & 0.0244 (0.00034) & 0.3238 (0.00025) & 0.4127 (0.00019) \\
\bottomrule
\end{tabular}
\end{table*}

The rarest state is particularly difficult. Under severe imbalance,
its unmixing RMSE is \(0.054\), with coverage \(0.951\). Prevalence
weighting reduces its contribution because its population share is
only \(0.10\).

Average posterior entropy, measured in nats, falls from \(0.639\)
under balance to \(0.486\) under severe imbalance, even as
\(\operatorname{cond}(M_0)\) rises. The majority-state prior makes
posteriors more concentrated, but observations informative about the
rare state become scarcer. Neither entropy nor MAP accuracy alone
therefore describes the precision of every state effect.

At \(n=50{,}000\), hard and soft true-effect coverage is 31.8\% and
31.9\% under balance, but zero for both methods under moderate and
severe imbalance. True-state oracle coverage is 94.1--95.1\% across
the three priors at this sample size. The hard and soft intervals
concentrate around mixture effects rather than the target state effects.

\section{Estimated-HMM simulation details}
\label{app:sim-hmm}

\subsection{Design and evaluation}
\label{app:sim-hmm-design}

This experiment adds HMM estimation error to the posterior ambiguity
studied above. We generate \(T=21\) days of pre-treatment activity,
action, and viewing intensity from a five-state joint HMM with the same
layers as the application model. Initial state probabilities are
uniform. Each transition row has probability 0.76 of remaining in the
same state and 0.06 of moving to each other state. The activity
intercepts are \(\thetab=(-0.8,-0.4,0,0.4,0.8)^\top\); each state has
a distinct modal action with probability 0.70; and the ordered
intensity coefficients are \((\kappa_s)_{s=1}^5=(-2,-1,0,1,2)\).

In each of 200 replications, we fit the correctly specified five-state
HMM by three-start EM on nested panels of
\(N\in\{300,600,1{,}200,2{,}400\}\) units. These HMM training
units are drawn from the target history distribution and are
independent of the randomized target panel. We analyze
prefixes of \(n\in\{2{,}000,10{,}000,50{,}000\}\) units from that
panel; each prefix supplies its own correction statistics and effect
moments. Another 10,000 target units are held out only for
conditional-effect prediction. Methods share HMM fits, treatment
assignments, and outcome errors within each replication.

Treatment is randomized with \(e=1/2\), and outcomes follow
\begin{equation*}
  Y_i=(A_i-1/2)\tau_{0,Z_{iT}}+\veps_i,
  \qquad \veps_i\sim\mathcal N(0,1).
\end{equation*}
The history-conditional mean outcome is zero. The primary effect profile
is \(\taub_0=(-1.2,-0.6,0,0.6,1.2)^\top\), ordered by the true
activity intercept. The nonmonotone and constant control profiles are
\((0.6,-1.2,1.2,-0.6,0)^\top\) and
\((0.6,0.6,0.6,0.6,0.6)^\top\), respectively.

We order fitted states by the rank of their estimated activity
intercept \(\widehat\theta_k\), without consulting treatment outcomes.
Holding the unit offset and weekday fixed, a larger intercept
implies a higher daily activity probability. We apply this permutation
to fitted parameters, posterior columns, scores, and derivatives; true
states use the same activity order. An outcome-blind match of true and
fitted target posteriors, used only as a diagnostic, agrees with this
rule in all 200 replications at the two primary HMM training sizes,
600 and 1,200. Median fitted-posterior MSE on the 50,000 target
histories, averaged over units and states, decreases from 0.001902
at \(N=300\) to 0.000245 at \(N=2{,}400\).

For method \(a\), replication \(r\), and HMM training size \(N\), define
equal-state MSE by
\begin{equation*}
  L_{r,N}^{(a)}=\frac15\sum_{k=1}^5
  (\widehat\tau_{r,N,k}^{(a)}-\tau_{0k})^2.
\end{equation*}
The primary contrast at \(n=50{,}000\) averages the paired differences
over \(N=600\) and \(N=1{,}200\):
\begin{equation*}
  D_r=\frac12\sum_{N\in\{600,1{,}200\}}
    \{L_{r,N}^{(\mathrm{plain})}
      -L_{r,N}^{(\mathrm{OPG})}\}.
\end{equation*}
Positive \(D_r\) favors correction. Its paired 95\% Monte Carlo
confidence interval (CI) uses the mean
\(\bar D=200^{-1}\sum_{r=1}^{200}D_r\), the
standard deviation \(s_D\) across replications, and the 97.5th
percentile \(t_{0.975,199}\) of a Student \(t\) distribution with 199
degrees of freedom: \(\bar D\pm t_{0.975,199}s_D/\sqrt{200}\).
This interval measures
precision of the mean performance gain, rather than coverage of a
state effect. Held-out conditional average treatment effect (CATE)
error compares
\(\qb_{\widehat\eta}(H)^\top\widehat{\taub}\) with
\(\qb_0(H)^\top\taub_0\), a target invariant to state labels.

\subsection{Orthogonal Correction}
\label{app:sim-hmm-correction}

Plain unmixing uses the fitted terminal posterior in
\eqref{eq:unmix-plugin}. For correction, the target-law sensitivities
\(G_0\), \(J_0\), and \(B_0=G_0J_0^{-1}\) are defined in
\eqref{eq:unmix-hmm-sensitivity}. At the true HMM, randomization and
the mixture identity give the following expression, where
\(\E_{\mathrm{target}}\) averages over target histories:
\begin{equation*}
  G_0=-v_A\E_{\mathrm{target}}\!\left[
    \qb_0(H)\{\taub_0^\top D_\eta\qb_{\eta_0}(H)\}
  \right],
  \qquad v_A=\frac14.
\end{equation*}
The derivative of the leading posterior factor has conditional mean
zero at \((\taub_0,\eta_0)\). We use the resulting history-only plug-in
estimate on the same target prefix used for effect estimation:
\begin{equation*}
  \widehat G_{\mathrm{tar}}
  =-v_A \E_n\!\left[
    \qb_{\widehat\eta}(H)
    \{(\widehat{\taub}^{\,\mathrm{pl}})^\top
      D_\eta\qb_{\widehat\eta}(H)\}
  \right].
\end{equation*}
Posterior derivatives are averaged over exactly the target
units whose state effects are estimated.

Treating unit offsets and weekday covariates as given, correct
conditional HMM specification implies a mean-zero observed-history
score \(\bm{s}^H(H;\eta)\), which integrates over latent paths.
Differentiating its conditional mean also differentiates the history
density. At the true parameter,
\begin{equation*}
  J_0=-\E_{\mathrm{target}}\!\left[
    \bm{s}^H(H;\eta_0)\bm{s}^H(H;\eta_0)^\top
  \right].
\end{equation*}
The primary outer-product-of-gradients (OPG) approximation
\citep{berndt1974estimation} uses scores from each analyzed target
prefix. Let \(I\) be the identity matrix of the score dimension:
\begin{equation*}
  \widehat J_{\mathrm{OPG}}
  =-\left\{\E_n[
      \bm{s}^H(H;\widehat\eta)
      \bm{s}^H(H;\widehat\eta)^\top]
      +10^{-4}I\right\},
  \qquad
  \widehat B_{\mathrm{OPG}}
  =\widehat G_{\mathrm{tar}}\widehat J_{\mathrm{OPG}}^{-1}.
\end{equation*}
The score Gram is positive semidefinite, and the ridge stabilizes
inversion without repeated score evaluations at perturbed HMM
parameters. We insert \(\widehat B_{\mathrm{OPG}}\) into
\eqref{eq:unmix-orthogonal-estimator}. A secondary arm uses the same
\(\widehat G_{\mathrm{tar}}\) and estimates \(J_0\) by central differences
of the mean score on the same target prefix, symmetrizes it, and uses
an outcome-blind ridge to cap the information condition number at
2,000. This arm was prespecified for \(N=600\) and \(N=1{,}200\);
its \(N=300\) and \(N=2{,}400\) results were added after the primary
run using the same seeds and panels.

The history-only \(\widehat G_{\mathrm{tar}}\) is a plug-in estimate of
the same population \(G_0\) targeted by the direct sample derivative
in \eqref{eq:unmix-sensitivity-estimator}. OPG uses the information
identity at the correctly specified true HMM. Fitted scores, finite
target data, and especially the fixed ridge make
\(\widehat B_{\mathrm{OPG}}\) a regularized approximation to \(B_0\).
The exact zero-derivative assertion of
Theorem~\ref{thm:orthogonal-unmixing} concerns unregularized \(B_0\);
the fixed-ridge implementation is assessed empirically here.
Misspecification can also invalidate the information identity.

A true-posterior oracle uses \(\qb_0(H_i)\) with the same target
outcomes; it does not observe each unit's realized latent state.
For the oracle, plain unmixing, OPG, and capped-Jacobian methods, we
report nominal 95\% confidence intervals of the form
\(\widehat\tau_k\pm1.96\widehat{\operatorname{SE}}_k\).
Standard errors use the empirical sandwich variance of each
target-unit moment, including the HMM score term for corrected
methods. They treat the fitted HMM and estimated correction coefficient
as fixed. Coverage averages the five state effects over 200
replications.

\subsection{Full-workflow state-effect results}
\label{app:sim-hmm-results}

For the well-conditioned five-state design, the tables below give all
four estimators at every HMM training size \(N\) and target size \(n\)
for the three effect profiles. Reported RMSE is the square root of the
mean equal-state MSE over 200 replications at each \((N,n)\).
Each method cell shows equal-state RMSE (Monte Carlo SE) above
empirical coverage (Monte Carlo SE) of nominal 95\% state-effect
intervals. Coverage averages five states within each of 200
replications. These Monte Carlo SEs describe uncertainty in the
simulation summaries; the
unit-level standard errors used to construct the intervals treat
the fitted HMM and correction coefficient as fixed. The oracle uses
the true terminal posterior and the same target outcomes, without
observing realized states.

\paragraph{Paired correction gains.}
Under monotone effects with \(n=50{,}000\), the primary paired contrast
averaging \(N=600\) and \(N=1{,}200\) within replication gives an OPG
state-MSE gain of \(0.502\times10^{-3}\), with paired 95\% Monte Carlo
interval \([0.413,0.591]\times10^{-3}\). This is a 31.4\% reduction
in mean state MSE, with lower MSE in 84.0\% of replications. The
corresponding capped-\(J\) gain is \(0.390\times10^{-3}\), interval
\([0.279,0.501]\times10^{-3}\), a 24.4\% reduction and a 69.0\% win
rate. Table~\ref{tab:full-workflow-paired-gains} separates the HMM
training sizes; positive gains favor correction.

\begin{table}[t]
\centering
\caption{Paired plain-minus-corrected equal-state MSE gain at
\(n=50{,}000\) under monotone effects, in units of \(10^{-3}\).
Parentheses give Monte Carlo SEs, brackets give paired 95\% Monte Carlo
intervals, and win rates give the fraction of 200 replications with
lower MSE after correction. Capped-\(J\) results at \(N=300\) and
\(N=2{,}400\) were added after the primary run.}
\label{tab:full-workflow-paired-gains}
\scriptsize
\setlength{\tabcolsep}{4pt}
\begin{tabular}{rcccc}
\toprule
\(N\) & OPG gain (SE) [95\% CI] & Capped \(J\) gain (SE) [95\% CI]
& OPG win & Capped \(J\) win \\
\midrule
300   & 1.321 (0.108) [1.107, 1.535] & 0.547 (0.198) [0.157, 0.937] & 89.5\% & 65.5\% \\
600   & 0.655 (0.056) [0.546, 0.765] & 0.468 (0.076) [0.317, 0.618] & 82.0\% & 70.5\% \\
1,200 & 0.348 (0.044) [0.262, 0.434] & 0.313 (0.053) [0.209, 0.416] & 76.0\% & 67.5\% \\
2,400 & 0.153 (0.029) [0.095, 0.210] & 0.151 (0.032) [0.088, 0.214] & 63.0\% & 61.0\% \\
\bottomrule
\end{tabular}
\end{table}

\paragraph{Posterior accuracy and score conditioning.}
Table~\ref{tab:full-workflow-posterior-accuracy} uses all 50,000
target histories for the posterior diagnostic and raw OPG condition
number. Posterior MSE falls as \(N\) grows, while raw OPG condition
numbers remain near \(10^3\); the separate numerical-Jacobian
condition cap also activates in many fits. All 800 HMM fits and
scheduled numerical-Jacobian calculations completed.

\begin{table}[t]
\centering
\caption{Fitted-HMM posterior accuracy and correction conditioning by
training size \(N\). Posterior MSE compares fitted and true terminal
posteriors on 50,000 target histories; entries are medians over 200
replications. Outcome-blind matching is a diagnostic and does not
choose reported state labels. Condition numbers use the raw target
OPG at \(n=50{,}000\). Parentheses are bootstrap Monte Carlo SEs of
medians from 2,000 resamples; mismatch and cap entries are counts out
of 200 fits.}
\label{tab:full-workflow-posterior-accuracy}
\scriptsize
\setlength{\tabcolsep}{4pt}
\begin{tabular}{rccccc}
\toprule
\(N\) & Posterior MSE & Matched MSE & Mismatches
& OPG cond. & Numerical cap active \\
\midrule
300   & 0.001902 (0.000065) & 0.001902 (0.000065) & 3/200 & 1,358 (62) & 198/200 \\
600   & 0.000947 (0.000027) & 0.000947 (0.000027) & 0/200 & 1,031 (11) & 173/200 \\
1,200 & 0.000472 (0.000014) & 0.000472 (0.000014) & 0/200 & 936 (7.0) & 75/200 \\
2,400 & 0.000245 (0.0000073) & 0.000245 (0.0000073) & 0/200 & 893 (2.8) & 11/200 \\
\bottomrule
\end{tabular}
\end{table}

\paragraph{Complete state-effect and coverage grids.}
Tables~\ref{tab:full-workflow-monotone-full}--\ref{tab:full-workflow-constant-full}
report every target prefix for the monotone, nonmonotone, and constant
effect profiles, respectively.

\begin{table}[t]
\centering
\caption{Monotone effects: equal-state RMSE and coverage of nominal
95\% state-effect intervals across all HMM and target sample sizes.
Each cell gives RMSE (Monte Carlo SE) above coverage (Monte Carlo SE)
over 200 replications.}
\label{tab:full-workflow-monotone-full}
\footnotesize
\setlength{\tabcolsep}{4pt}
\begin{tabular}{rrcccc}
\toprule
\(N\) & \(n\) & Oracle & Plain & OPG & Capped \(J\) \\
\midrule
300 & 2,000 & \shortstack{0.15773 (0.0034)\\0.949 (0.0071)} & \shortstack{0.17206 (0.0058)\\0.927 (0.0088)} & \shortstack{0.16861 (0.0059)\\0.938 (0.0082)} & \shortstack{0.16686 (0.0059)\\0.945 (0.0079)} \\
& 10,000 & \shortstack{0.07231 (0.0018)\\0.947 (0.0077)} & \shortstack{0.09468 (0.0065)\\0.896 (0.0097)} & \shortstack{0.08736 (0.0070)\\0.927 (0.0084)} & \shortstack{0.09033 (0.0066)\\0.918 (0.0089)} \\
& 50,000 & \shortstack{0.03248 (0.00083)\\0.936 (0.0086)} & \shortstack{0.06722 (0.0083)\\0.771 (0.013)} & \shortstack{0.05655 (0.0098)\\0.899 (0.011)} & \shortstack{0.06302 (0.0087)\\0.849 (0.012)} \\
\midrule
600 & 2,000 & \shortstack{0.15773 (0.0034)\\0.949 (0.0071)} & \shortstack{0.16204 (0.0036)\\0.943 (0.0071)} & \shortstack{0.15947 (0.0036)\\0.948 (0.0070)} & \shortstack{0.15892 (0.0036)\\0.953 (0.0068)} \\
& 10,000 & \shortstack{0.07231 (0.0018)\\0.947 (0.0077)} & \shortstack{0.07799 (0.0019)\\0.926 (0.0087)} & \shortstack{0.07361 (0.0019)\\0.939 (0.0077)} & \shortstack{0.07433 (0.0020)\\0.942 (0.0076)} \\
& 50,000 & \shortstack{0.03248 (0.00083)\\0.936 (0.0086)} & \shortstack{0.04200 (0.0010)\\0.859 (0.011)} & \shortstack{0.03330 (0.00086)\\0.930 (0.0086)} & \shortstack{0.03601 (0.0011)\\0.919 (0.0092)} \\
\midrule
1,200 & 2,000 & \shortstack{0.15773 (0.0034)\\0.949 (0.0071)} & \shortstack{0.16065 (0.0035)\\0.941 (0.0073)} & \shortstack{0.15888 (0.0035)\\0.949 (0.0069)} & \shortstack{0.15817 (0.0036)\\0.951 (0.0069)} \\
& 10,000 & \shortstack{0.07231 (0.0018)\\0.947 (0.0077)} & \shortstack{0.07604 (0.0019)\\0.929 (0.0086)} & \shortstack{0.07309 (0.0018)\\0.947 (0.0074)} & \shortstack{0.07303 (0.0019)\\0.945 (0.0076)} \\
& 50,000 & \shortstack{0.03248 (0.00083)\\0.936 (0.0086)} & \shortstack{0.03788 (0.00096)\\0.905 (0.0097)} & \shortstack{0.03296 (0.00087)\\0.930 (0.0088)} & \shortstack{0.03350 (0.00091)\\0.934 (0.0083)} \\
\midrule
2,400 & 2,000 & \shortstack{0.15773 (0.0034)\\0.949 (0.0071)} & \shortstack{0.15798 (0.0034)\\0.950 (0.0069)} & \shortstack{0.15824 (0.0035)\\0.953 (0.0066)} & \shortstack{0.15803 (0.0034)\\0.949 (0.0069)} \\
& 10,000 & \shortstack{0.07231 (0.0018)\\0.947 (0.0077)} & \shortstack{0.07357 (0.0018)\\0.945 (0.0076)} & \shortstack{0.07281 (0.0018)\\0.946 (0.0077)} & \shortstack{0.07277 (0.0019)\\0.948 (0.0075)} \\
& 50,000 & \shortstack{0.03248 (0.00083)\\0.936 (0.0086)} & \shortstack{0.03505 (0.00089)\\0.921 (0.0087)} & \shortstack{0.03280 (0.00086)\\0.938 (0.0084)} & \shortstack{0.03283 (0.00086)\\0.936 (0.0082)} \\
\bottomrule
\end{tabular}
\end{table}

\begin{table}[t]
\centering
\caption{Nonmonotone effects: equal-state RMSE and coverage of nominal
95\% state-effect intervals across all HMM and target sample sizes.
Each cell gives RMSE (Monte Carlo SE) above coverage (Monte Carlo SE)
over 200 replications.}
\label{tab:full-workflow-permuted-full}
\footnotesize
\setlength{\tabcolsep}{4pt}
\begin{tabular}{rrcccc}
\toprule
\(N\) & \(n\) & Oracle & Plain & OPG & Capped \(J\) \\
\midrule
300 & 2,000 & \shortstack{0.16494 (0.0036)\\0.942 (0.0073)} & \shortstack{0.22614 (0.023)\\0.911 (0.0096)} & \shortstack{0.21612 (0.023)\\0.933 (0.0083)} & \shortstack{0.26279 (0.019)\\0.893 (0.010)} \\
& 10,000 & \shortstack{0.07294 (0.0018)\\0.947 (0.0077)} & \shortstack{0.18018 (0.032)\\0.808 (0.014)} & \shortstack{0.16245 (0.035)\\0.914 (0.0097)} & \shortstack{0.20981 (0.025)\\0.772 (0.015)} \\
& 50,000 & \shortstack{0.03202 (0.00081)\\0.946 (0.0073)} & \shortstack{0.16850 (0.034)\\0.570 (0.014)} & \shortstack{0.14917 (0.038)\\0.837 (0.012)} & \shortstack{0.19231 (0.027)\\0.560 (0.017)} \\
\midrule
600 & 2,000 & \shortstack{0.16494 (0.0036)\\0.942 (0.0073)} & \shortstack{0.17595 (0.0038)\\0.930 (0.0081)} & \shortstack{0.16931 (0.0037)\\0.940 (0.0074)} & \shortstack{0.20242 (0.0074)\\0.914 (0.0086)} \\
& 10,000 & \shortstack{0.07294 (0.0018)\\0.947 (0.0077)} & \shortstack{0.09536 (0.0025)\\0.870 (0.011)} & \shortstack{0.07641 (0.0019)\\0.938 (0.0083)} & \shortstack{0.11388 (0.0049)\\0.852 (0.012)} \\
& 50,000 & \shortstack{0.03202 (0.00081)\\0.946 (0.0073)} & \shortstack{0.06857 (0.0018)\\0.711 (0.015)} & \shortstack{0.03588 (0.00091)\\0.927 (0.0082)} & \shortstack{0.08649 (0.0061)\\0.719 (0.017)} \\
\midrule
1,200 & 2,000 & \shortstack{0.16494 (0.0036)\\0.942 (0.0073)} & \shortstack{0.16908 (0.0036)\\0.939 (0.0075)} & \shortstack{0.16866 (0.0036)\\0.944 (0.0073)} & \shortstack{0.18139 (0.0044)\\0.938 (0.0077)} \\
& 10,000 & \shortstack{0.07294 (0.0018)\\0.947 (0.0077)} & \shortstack{0.08347 (0.0020)\\0.908 (0.0095)} & \shortstack{0.07508 (0.0019)\\0.945 (0.0076)} & \shortstack{0.08485 (0.0023)\\0.920 (0.0094)} \\
& 50,000 & \shortstack{0.03202 (0.00081)\\0.946 (0.0073)} & \shortstack{0.05195 (0.0012)\\0.794 (0.013)} & \shortstack{0.03357 (0.00084)\\0.938 (0.0077)} & \shortstack{0.04690 (0.0020)\\0.870 (0.012)} \\
\midrule
2,400 & 2,000 & \shortstack{0.16494 (0.0036)\\0.942 (0.0073)} & \shortstack{0.16665 (0.0035)\\0.942 (0.0076)} & \shortstack{0.16815 (0.0035)\\0.946 (0.0072)} & \shortstack{0.17160 (0.0038)\\0.949 (0.0069)} \\
& 10,000 & \shortstack{0.07294 (0.0018)\\0.947 (0.0077)} & \shortstack{0.07851 (0.0020)\\0.928 (0.0083)} & \shortstack{0.07459 (0.0019)\\0.949 (0.0074)} & \shortstack{0.07677 (0.0019)\\0.943 (0.0082)} \\
& 50,000 & \shortstack{0.03202 (0.00081)\\0.946 (0.0073)} & \shortstack{0.04339 (0.0011)\\0.864 (0.011)} & \shortstack{0.03310 (0.00084)\\0.940 (0.0076)} & \shortstack{0.03527 (0.00090)\\0.933 (0.0079)} \\
\bottomrule
\end{tabular}
\end{table}

\begin{table}[t]
\centering
\caption{Constant effects: equal-state RMSE and coverage of nominal
95\% state-effect intervals across all HMM and target sample sizes.
Each cell gives RMSE (Monte Carlo SE) above coverage (Monte Carlo SE)
over 200 replications.}
\label{tab:full-workflow-constant-full}
\footnotesize
\setlength{\tabcolsep}{4pt}
\begin{tabular}{rrcccc}
\toprule
\(N\) & \(n\) & Oracle & Plain & OPG & Capped \(J\) \\
\midrule
300 & 2,000 & \shortstack{0.15451 (0.0033)\\0.948 (0.0073)} & \shortstack{0.15481 (0.0035)\\0.947 (0.0073)} & \shortstack{0.15580 (0.0035)\\0.944 (0.0073)} & \shortstack{0.15334 (0.0034)\\0.950 (0.0069)} \\
& 10,000 & \shortstack{0.07057 (0.0018)\\0.948 (0.0075)} & \shortstack{0.07061 (0.0019)\\0.949 (0.0075)} & \shortstack{0.07110 (0.0019)\\0.943 (0.0077)} & \shortstack{0.07051 (0.0019)\\0.942 (0.0078)} \\
& 50,000 & \shortstack{0.03122 (0.00079)\\0.942 (0.0080)} & \shortstack{0.03120 (0.00076)\\0.937 (0.0080)} & \shortstack{0.03142 (0.00077)\\0.936 (0.0079)} & \shortstack{0.03102 (0.00075)\\0.941 (0.0079)} \\
\midrule
600 & 2,000 & \shortstack{0.15451 (0.0033)\\0.948 (0.0073)} & \shortstack{0.15509 (0.0035)\\0.945 (0.0075)} & \shortstack{0.15537 (0.0034)\\0.948 (0.0073)} & \shortstack{0.15460 (0.0034)\\0.950 (0.0071)} \\
& 10,000 & \shortstack{0.07057 (0.0018)\\0.948 (0.0075)} & \shortstack{0.07095 (0.0018)\\0.947 (0.0077)} & \shortstack{0.07117 (0.0019)\\0.947 (0.0076)} & \shortstack{0.07084 (0.0019)\\0.949 (0.0075)} \\
& 50,000 & \shortstack{0.03122 (0.00079)\\0.942 (0.0080)} & \shortstack{0.03118 (0.00078)\\0.937 (0.0083)} & \shortstack{0.03132 (0.00078)\\0.939 (0.0083)} & \shortstack{0.03107 (0.00078)\\0.941 (0.0082)} \\
\midrule
1,200 & 2,000 & \shortstack{0.15451 (0.0033)\\0.948 (0.0073)} & \shortstack{0.15499 (0.0034)\\0.947 (0.0074)} & \shortstack{0.15520 (0.0034)\\0.949 (0.0072)} & \shortstack{0.15505 (0.0034)\\0.953 (0.0066)} \\
& 10,000 & \shortstack{0.07057 (0.0018)\\0.948 (0.0075)} & \shortstack{0.07072 (0.0018)\\0.948 (0.0078)} & \shortstack{0.07084 (0.0018)\\0.945 (0.0078)} & \shortstack{0.07070 (0.0018)\\0.946 (0.0077)} \\
& 50,000 & \shortstack{0.03122 (0.00079)\\0.942 (0.0080)} & \shortstack{0.03129 (0.00080)\\0.939 (0.0083)} & \shortstack{0.03138 (0.00080)\\0.939 (0.0082)} & \shortstack{0.03129 (0.00081)\\0.942 (0.0080)} \\
\midrule
2,400 & 2,000 & \shortstack{0.15451 (0.0033)\\0.948 (0.0073)} & \shortstack{0.15440 (0.0033)\\0.949 (0.0072)} & \shortstack{0.15479 (0.0033)\\0.949 (0.0074)} & \shortstack{0.15450 (0.0033)\\0.949 (0.0071)} \\
& 10,000 & \shortstack{0.07057 (0.0018)\\0.948 (0.0075)} & \shortstack{0.07058 (0.0018)\\0.949 (0.0074)} & \shortstack{0.07067 (0.0018)\\0.945 (0.0078)} & \shortstack{0.07054 (0.0018)\\0.945 (0.0078)} \\
& 50,000 & \shortstack{0.03122 (0.00079)\\0.942 (0.0080)} & \shortstack{0.03126 (0.00080)\\0.938 (0.0083)} & \shortstack{0.03131 (0.00080)\\0.941 (0.0081)} & \shortstack{0.03126 (0.00080)\\0.941 (0.0081)} \\
\bottomrule
\end{tabular}
\end{table}

Across the monotone grid, OPG gains grow with the target sample size
at each fixed \(N\), as Figure~\ref{fig:orthogonal-opg-grid} shows.
Posterior error also decreases with HMM training size, consistent with
less first-stage error to correct. Under nonmonotone effects, OPG still
improves state RMSE at large target sizes, whereas capped \(J\) can
worsen it. Under constant effects, all four methods have similar RMSE.

\subsection{Prediction, effect profiles, and conditioning}
\label{app:sim-hmm-checks}

At the two primary HMM training sizes, correction does not improve
held-out CATE prediction. At \(N=600\), CATE RMSE is 0.04166 for
plain unmixing, 0.04534 for OPG, and 0.04835 for capped \(J\); at
\(N=1{,}200\), the corresponding
values are 0.03288, 0.03588, and 0.03709. Under nonmonotone effects,
the pooled OPG state-MSE gain is
\(2.493\times10^{-3}\), with paired interval
\([2.195,2.792]\times10^{-3}\). Under constant effects, its gain is
the small negative value \(-0.0072\times10^{-3}\), with interval
\([-0.0119,-0.0026]\times10^{-3}\). For capped \(J\), the pooled
nonmonotone gain is \(-1.140\times10^{-3}\), with interval
\([-2.266,-0.014]\times10^{-3}\); its constant-effect gain is near
zero. Thus the capped-Jacobian benefit is not uniform over effect
profiles.

Weak score directions can amplify error through \(B_0=G_0J_0^{-1}\).
At 50,000 target units, median condition numbers of the raw OPG
score Gram are 1,031 and 936 for \(N=600\) and \(N=1{,}200\);
the numerical-Jacobian condition cap activates in 173 and 75 of
200 fits, respectively. The smallest eigenvalue and condition number
of the unregularized target score information, together with the ridge
used for inversion, are useful diagnostics. The reported gains concern
regularized correction and do not evaluate a rule for withholding it
when score information is weak.

The HMM family and state count are correctly specified, treatment is
randomized with known propensity, and the history-conditional mean
outcome is zero by design. Exploratory pilots informed the
data-generating process. The same target units estimate the
correction and solve the effect moment; reported unit-level
standard errors hold the fitted HMM and correction coefficient fixed.
These intervals have no general unconditional 95\% coverage
guarantee. The experiments do not address HMM misspecification,
observational treatment, or alternative regularization rules.

\clearpage
\section{Additional case-study figures}
\label{app:case-study-figures}

\begin{figure}[htbp]
\centering
\includegraphics[width=\textwidth]{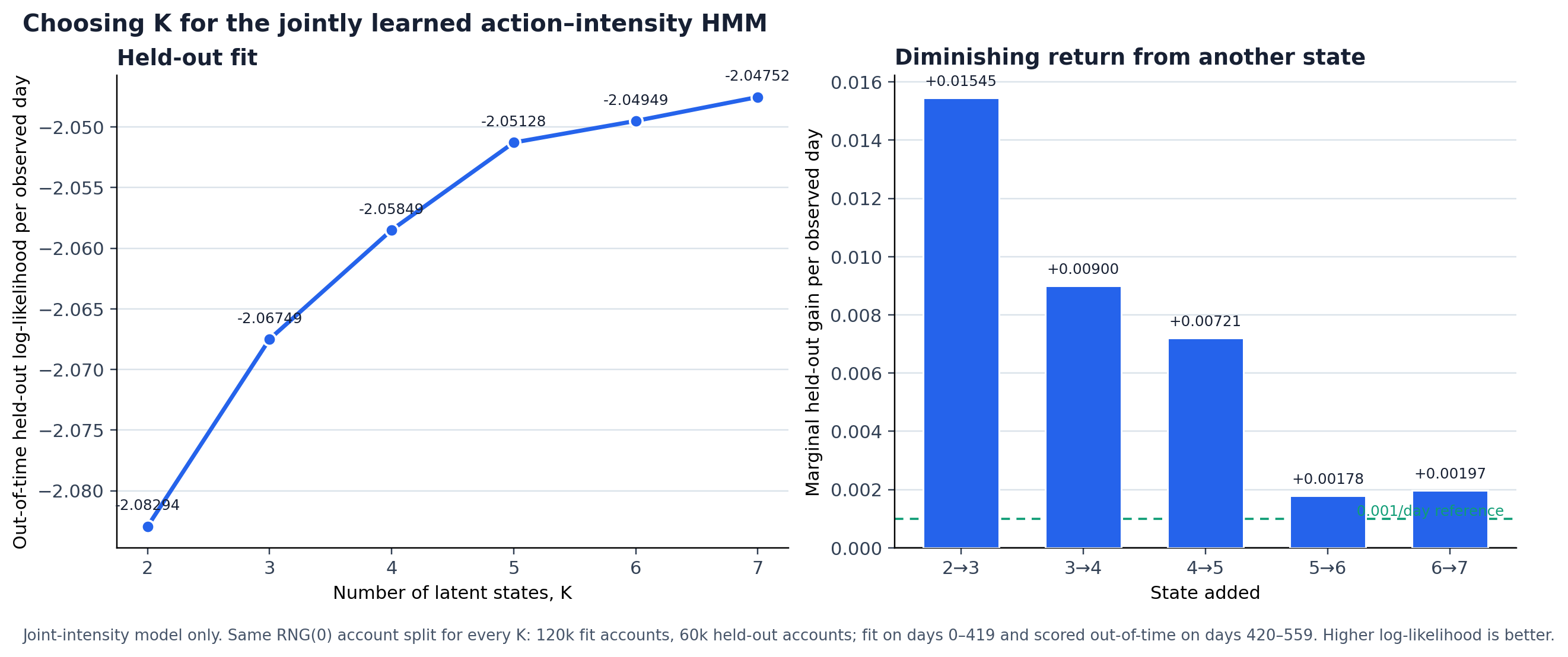}
\caption{\small Selection of the HMM state count using out-of-time
held-out log-likelihood per observed day. Candidate joint
action--intensity HMMs use the same 120,000 fitting units and 60,000
held-out units across state counts. Models are fitted on days 0--419
and evaluated on days 420--559. The right panel shows the gain from
adding one state, which drops sharply after \(K=5\).}
\label{fig:hmm-k-selection}
\end{figure}

\begin{figure}[htbp]
\centering
\includegraphics[width=\textwidth]{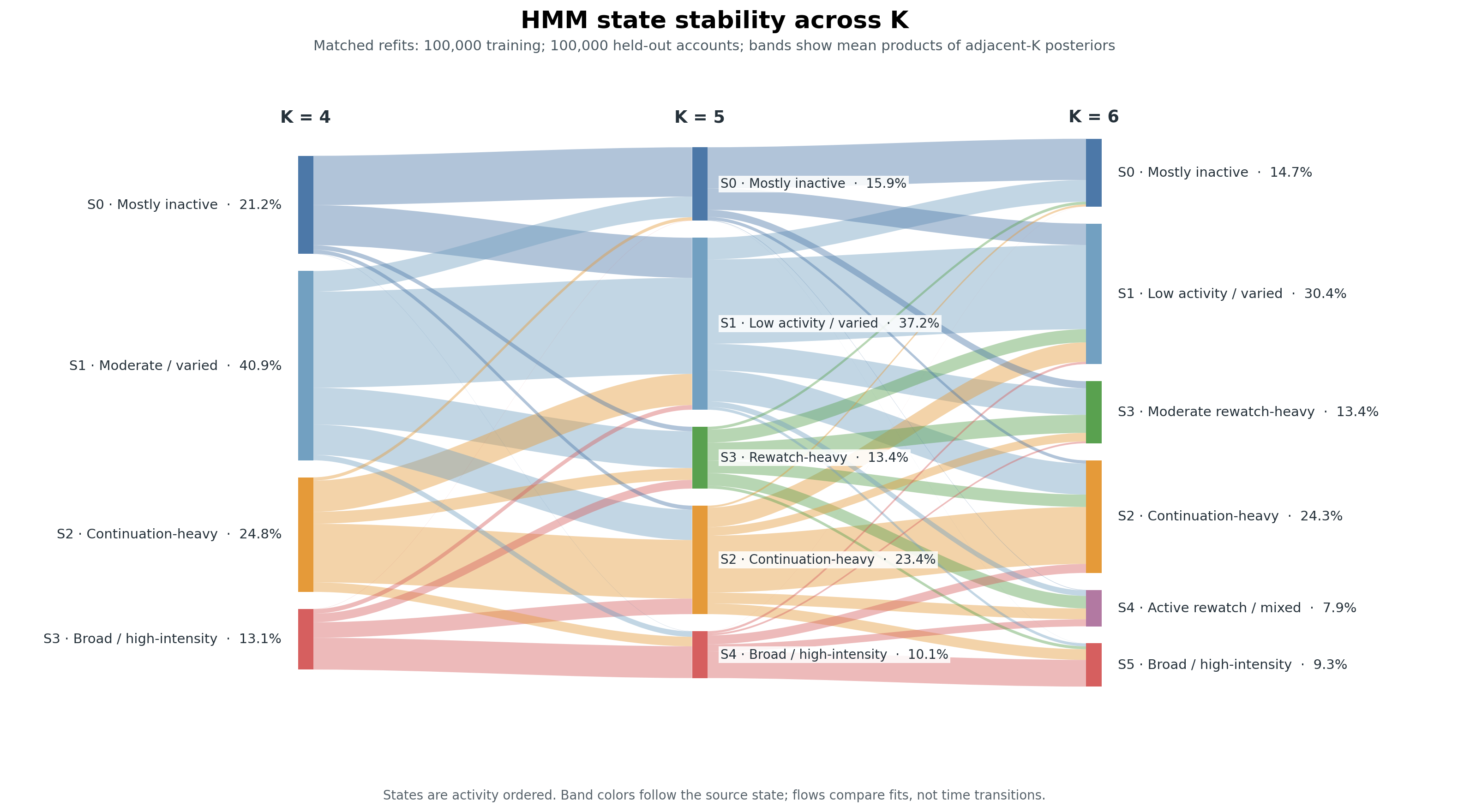}
\caption{\small Stability of behavioral-state posteriors across HMM
fits with \(K=4,5,6\). The fits use 100,000 training accounts and are
compared on the same 100,000 held-out accounts. Node percentages are
mean posterior state probabilities on the held-out accounts; each
band's width is proportional to the mean product of its endpoint
posterior probabilities for adjacent state counts. States are ordered
by activity. Bands compare fitted models and do not represent
transitions over time.}
\label{fig:hmm-state-stability}
\end{figure}

\begin{figure}[htbp]
\centering
\includegraphics[width=\textwidth]{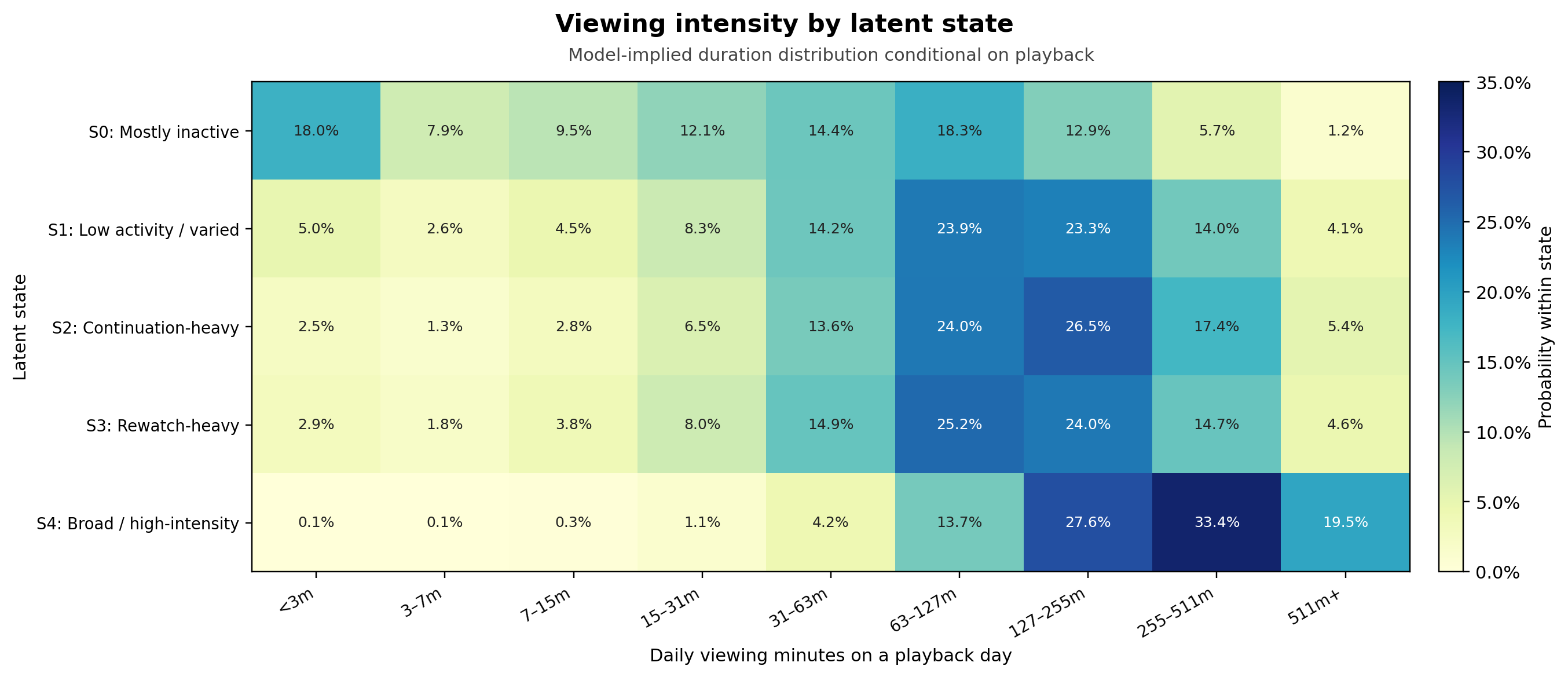}
\caption{\small Model-implied distribution of daily viewing duration
on a playback day, conditional on each state of the fitted five-state
HMM. Each row gives the probability of a viewing-minute bin and sums
to 100\%. The behavioral profile names correspond to those in
Figure~\ref{fig:state_action}; the display order differs.}
\label{fig:state-intensity}
\end{figure}
\end{document}